\documentclass[twoside]{IEEEtran}
\usepackage{url,epsfig}
\usepackage{amsmath, amsthm, amssymb}
\usepackage{multirow}
\usepackage{amsfonts}
\usepackage{graphicx}
\usepackage{setspace}
\usepackage{datetime}
\usepackage{multirow}

\usepackage{soul}
\usepackage{rotating,cite}
\newcommand{\bfm}[1]{\mbox{\boldmath $#1$}}
\newcommand*\xbar[1]{%
  \hbox{%
    \vbox{%
      \hrule height 0.45pt 
      \kern0.5ex
      \hbox{%
        \kern-0.1em
        \ensuremath{#1}%
        \kern-0.1em
      }%
    }%
  }%
}

\newtheorem{thm}{Theorem}
\newtheorem{cor}{Corollary}
\newtheorem{lem}{Lemma}
\newtheorem{prop}{Proposition}
\newtheorem{claim}{Claim}
\date{\today}
\begin{document}

\title{Bounds on the Capacity of Random Insertion and Deletion-Additive Noise Channels}

\author{Mojtaba Rahmati and Tolga M.~Duman,~\IEEEmembership{Fellow,~IEEE}\thanks{Manuscript received January~5, 2011; revised June~12, 2012 and February~20, 2013; accepted March~31, 2013. This research is funded by the National Science Foundation under contract NSF-TF 0830611.}\thanks{M. Rahmati is with the School of Electrical, Computer and Energy Engineering (ECEE) of Arizona State University, Tempe, AZ 85287-5706, USA (email: mojtaba@asu.edu); T.~M.~Duman is with the Department of Electrical and Electronics Engineering, Bilkent University, Bilkent, Ankara, 06800, Turkey (email: duman@ee.bilkent.edu.tr) and he is on leave from the School of ECEE of Arizona State University.}\thanks{This paper was presented in part at IEEE GLOBECOM 2011, Houston, TX, Dec.~2011.}\thanks{Copyright (c) 2013 IEEE. Personal use of this material is permitted. However, permission to use this material for any other purposes must be obtained from the IEEE by sending a request to pubs-permissions@ieee.org.}}
\maketitle
\begin{abstract}
We develop several analytical lower bounds on the capacity of binary insertion and deletion channels by considering independent uniformly distributed (i.u.d.) inputs and computing lower bounds on the mutual information between the input and output sequences. For the deletion channel, we consider two different models: independent and identically distributed (i.i.d.) deletion-substitution channel and i.i.d. deletion channel with additive white Gaussian noise (AWGN). These two models are considered to incorporate effects of the channel noise along with the synchronization errors. For the insertion channel case we consider the Gallager's model in which the transmitted bits are replaced with two random bits and uniform over the four possibilities independently of any other insertion events. The general approach taken is similar in all cases, however the specific computations differ. Furthermore, the approach yields a useful lower bound on the capacity for a wide range of deletion probabilities for the deletion channels, while it provides a beneficial bound only for small insertion probabilities (less than 0.25) for the insertion model adopted. We emphasize the importance of these results by noting that 1) our results are the first analytical bounds on the capacity of deletion-AWGN channels, 2) the results developed are the best available analytical lower bounds on the deletion-substitution case, 3) for the Gallager insertion channel model, the new lower bound improves the existing results for small insertion probabilities.
\end{abstract}

\noindent
\begin{IEEEkeywords}
Insertion/deletion channels, synchronization, channel capacity, achievable rates.
\end{IEEEkeywords}


\section{Introduction}
In modeling digital communication systems, we often assume that the transmitter and receiver are completely synchronized; however, achieving a perfect time-alignment between the transmitter and receiver clocks is not possible in all communication systems and synchronization errors are unavoidable. A useful model for synchronization errors assumes that the number of received bits may be more or less than the number of transmitted bits. In other words, insertion/deletion channels may be used as appropriate models for communication channels that suffer from synchronization errors. Due to the memory introduced by the synchronization errors, an information theoretic study of these channels proves to be very challenging. For instance, even for seemingly simple models such as an i.i.d. deletion channel, an exact calculation of the capacity is not possible and only upper/lower bounds (which are often loose) are available.

In this paper, we compute analytical lower bounds on the capacity of the i.i.d. deletion channel with substitution errors and in the presence of AWGN, and i.i.d. random insertion channel, by lower bounding the mutual information rate between the transmitted and received sequences for i.u.d. inputs. We particularly focus on the small insertion/deletion probabilities with the premise that such small values are more practical from an application point of view, where every bit is independently deleted with probability $p_d$ or replaced with two randomly chosen bits with probability $p_i$, while neither the transmitter nor the receiver have any information about the positions of deletions and insertions, and undeleted bits are flipped with probability $p_e$ and bits are received in the correct order. By a deletion-substitution channel we refer to an insertion/deletion channel with $p_i=0$; by a deletion-AWGN channel we refer to an insertion/deletion channel with $p_i=p_e=0$ (deletion-only channel) in which undeleted bits are received in the presence of AWGN, that can be modeled by a combination of a deletion-only channel with a binary input AWGN (BI-AWGN) channel such that every bit first goes through a  deletion-only channel and then through a BI-AWGN channel. Finally, by a random insertion channel we refer to an insertion/deletion channel with $p_d=p_e=0$.

\subsection{Review of Existing Results}
Dobrushin~\cite{dobrushin} proved under very general conditions that for a memoryless channel with synchronization errors, Shannon's theorem on transmission rates applies and the information and transmission capacities are equal. The proof hinges on showing that information stability holds for the insertion/deletion channels and, as a result~\cite{dobrushin_general}, capacity per bit of an i.i.d. insertion/deletion channel can be obtained by $\displaystyle \lim_{N\to \infty}\displaystyle \max_{P\left(\bfm{X}\right)}\dfrac{1}{N}I(\bfm{X};\bfm{Y})$, where $\bfm X$ and $\bfm Y$ are the transmitted and received sequences, respectively, and $N$ is the length of the transmitted sequence. On the other hand, there is no single-letter or finite-letter formulation which may be amenable for the capacity computation, and no results are available providing the   exact value of the limit.

Gallager~\cite{gallager} considered the use of convolutional codes over channels with synchronization errors, and derived an expression which represents an achievable rate for channels with insertion, deletion and substitution errors (whose model is specified earlier). The approach is to consider transmission of i.u.d. binary information sequences by convolutional coding and modulo-2 addition of a pseudo-random binary sequence (which could be considered as a watermark used for synchronization purposes), and computation of a rate that guarantees a successful decoding by sequential decoding. The achievable rate, or the capacity lower bound, is given by the expression
  \begin{equation}  \label{eq:LB-gallager}
  C \geq 1+p_d \log{p_d}+p_i \log{p_i}+p_c
  \log{p_c}+p_s\log{p_s},
  \end{equation}
where $C$ is the channel capacity, $p_c=(1-p_d-p_i)(1-p_e)$ is the probability of correct reception, and $p_s=(1-p_d-p_i)p_e$ is the probability that a flipped version of the transmitted bit is received. The logarithm is taken base 2 resulting in transmission rates in bits/channel use. By substituting $p_i=0$ in Eq.~\eqref{eq:LB-gallager}, for $p_d\leq0.5$, a lower bound on the capacity of the deletion-substitution channel $C_{ds}$, can be obtained as
\begin{equation}\label{LB_gallager_del_sub}
C_{ds}\geq 1-H_b(p_d)-(1-p_d)H_b(p_e),
\end{equation}
where $H_b(p_d)=-p_d\log {p_d}-(1-p_d)\log(1-p_d)$ is the binary entropy function. It is interesting to note that for $p_d=p_e=0$ ($p_i=p_e=0$) and $p_i\le 0.5$ ($p_d\le 0.5$), a lower bound on the capacity of the random insertion channel (deletion-only channel) with insertion (deletion) probability of $p_i$ ($p_d$), is equal to the capacity of a binary symmetric channel (BSC) with a substitution error probability of $p_i$ ($p_d$).

In~\cite{diggavi2001transmission,diggavi2006information}, authors argue that, since the deletion channel has memory, optimal codebooks for use over deletion channels should have memory. Therefore, in~\cite{diggavi2001transmission,diggavi2006information,drinea2006lower,drinea2007improved}, achievable rates are computed by using a random codebook of rate $R$ with $2^{n\cdot R}$ codewords of length $n$, while each codeword is generated independently according to a symmetric first-order Markov process. Then, the generated codebook is used for transmission over the i.i.d. deletion channel. In the receiver, different decoding algorithms are proposed, e.g., in~\cite{diggavi2001transmission}, if the number of codewords in the codebook that contain the received sequence as a subsequence is only one, the transmission is successful, otherwise an error is declared. The proposed decoding algorithms result in an upper bound for the incorrect decoding probability. Finally, the maximum value of $R$ that results in a successful decoding as $n\rightarrow\infty$ is an achievable rate, hence a lower bound on the transmission capacity of the deletion channel. The lower bound~\eqref{eq:LB-gallager}, for $p_i=p_e=0$, is also proved in~\cite{diggavi2001transmission} using a different approach compared to the one taken by Gallager~\cite{gallager}, where the authors computed achievable rates by choosing codewords randomly, independently and uniformly among all possible codewords of a certain length.

  In~\cite{drinea2007}, a lower bound on the capacity of the deletion channel is directly obtained by lower bounding the information capacity $\displaystyle \lim_{N \to\infty}\dfrac{1}{N}\max_{P(\bfm{X})}I(\bfm{X};\bfm{Y})$. In~\cite{drinea2007}, input sequences are considered as alternating blocks of zeros and ones (runs), where the length of the runs $L$ are i.i.d. random variables following a particular distribution over positive integers with a finite expectation and finite entropy ($E(L), H(L)<\infty$ where $E(\cdot)$ and $H(\cdot$) denote the expected value and entropy, respectively).

In~\cite{kavcic2004insertion,junhu}, Monte Carlo methods are used for computing lower bounds on the capacity of the insertion/deletion channels based on reduced-state techniques. In~\cite{kavcic2004insertion}, the input process is assumed to be a stationary Markov process and lower bounds on the capacity of the deletion and insertion channels are obtained via Monte Carlo simulations considering both the first and second-order Markov processes as input. In~\cite{junhu}, information rates for i.u.d. input sequences are computed for several channel models using a similar Monte Carlo approach where in addition to the insertions/deletions, effects of intersymbol interference (ISI) and AWGN are also investigated.

  There are several papers deriving upper bounds on the capacity of the insertion/deletion channels as well. Fertonani and Duman in~\cite{dario} present several novel upper bounds on the capacity of the i.i.d. deletion channel by providing the decoder (and possibly the encoder) with some genie-aided information about the deletion process resulting in auxiliary channels whose capacities are certainly upper bounds on the capacity of the i.i.d. deletion channel. By providing the decoder with appropriate side information, a memoryless channel is obtained in such a way that Blahut-Arimoto algorithm (BAA) can be used for evaluating the capacity of the auxiliary channels (or, at least computing a provable upper bound on their capacities). They also prove that by subtracting some value from the derived upper bounds, lower bounds on the capacity can be derived. The intuition is that the subtracted information is more than extra information added by revealing certain aspects of the deletion process. A nontrivial upper bound on the deletion channel capacity is also obtained in~\cite{diggavi-capacity} where a different genie-aided decoder is considered. Furthermore, Fertonani and Duman in~\cite{dario2} extend their work~\cite{dario} to compute several upper and lower bounds on the capacity of channels with insertion, deletion and substitution errors as well.

  In two recent papers~\cite{kanoria,asymptotic}, asymptotic capacity expressions for the binary i.i.d. deletion channel for small deletion probabilities are developed. In~\cite{asymptotic}, the authors prove that $C_d \leq 1-(1-O(p_d))H_b(p_d)$ (where $O(.)$ represents the standard Landau (big-O) notation) which clearly shows that for small deletion probabilities, $1-H_b(p_d)$ is a tight lower bound on the capacity of the deletion channel. In~\cite{kanoria}, an expansion of the capacity for small deletion probabilities is computed with several dominant terms in an explicit form. The interpretation of our result for i.i.d. deletion-only channel case is parallel to the one in~\cite{asymptotic}.

\subsection{Contributions of the Paper}

In this paper, we focus on small insertion/deletion probabilities and derive analytical lower bounds on the capacity of the insertion/deletion channels by lower bounding the mutual information between i.u.d. input sequences and resulting output sequences. Since as shown in~\cite{dobrushin}, for an insertion/deletion channel, the information and transmission capacities are equal justifying our approach in obtaining an achievable rate.

We note that our idea is somewhat similar to the idea of directly lower bounding the information capacity instead of lower bounding the transmission capacity as employed in~\cite{drinea2007}. However, there are fundamental differences in the main methodology as will become apparent later. For instance, our approach provides a procedure that can easily be employed for many different channel models with synchronization errors as such we are able to consider deletion-substitution, deletion-AWGN and random insertion channels. Other differences include adopting a finite-length transmission which is proved to yield a lower bound on the capacity after subtracting some appropriate term, and the complexity in computing the final expression numerically is much lower in many versions of our results.

Finally, we emphasize that by utilizing the new approach, we improve upon the obtained results in the existing literature in several different aspects. In particular, the contributions of the paper include
\begin{itemize}
\item development of a new approach for deriving achievable information rates for insertion/deletion channels, 
\item the first analytical lower bound on the capacity of the deletion-AWGN channel,
\item tighter analytical lower bounds on the capacity of the deletion-substitution channel for all values of deletion and substitution probabilities compared to the existing analytical results,
\item tighter analytical lower bounds on the capacity of the random insertion channels for small values of insertion probabilities ($p_i<0.25$) compared to the existing lower bounds,
\item very simple lower bounds on the capacity of several cases of insertion/deletion channels.
\end{itemize}

\noindent Regarding the final point, we note that by employing $p_e=0$ in the results on the deletion-substitution channel, we arrive at lower bounds on the capacity of the deletion-only channel which are in agreement with the asymptotic results of~\cite{kanoria,asymptotic} in the sense of capturing the dominant terms in the capacity expansion. Our results, however, are provable lower bounds on the capacity, while the existing asymptotic results are not amenable for numerical calculation (as they contain big-O terms). 

\subsection{Notation}
We denote a binary sequence of length $n$ with $K$ runs by $(b;n_1,n_2,\dotsc  ,n_K)$, where $b\in\{0,1\}$ denotes the first run type and $\sum_{k=1}^Kn_k=n$. For example, the sequence 001111011000 can be represented as (0;2,4,1,2,3). We use four different ways to denote different sequences; $\bfm x(b;n^{x};K^{x})$ represents every sequence belonging to the set of sequences of length $n^{x}$ with $K^{x}$ runs and by the first run of type $b$, $\bfm x(b;n^{x};K^{x};l)$ represents a sequence $\bfm x(b;n^{x};K^{x})$ which has $l$ runs of length one ($l=\sum_{k=1}^{K^x} \delta(n_k^x-1)$ with $\delta(.)$ denoting the Kronecker delta function), $\bfm x(n^{x})$ represents every sequence of length $n^{x}$, and $\bfm x$ represents every possible sequence. The set of all input sequences is shown by $\cal X$, and the set of output sequences of the deletion-only, and random insertion channels are shown by ${\cal Y}^d$ and ${\cal Y}^i$, respectively. ${\cal{Y}}^d_{-a}$ and ${\cal{Y}}^i_{+c}$ denote the set of output sequences resulting from $a$ deletions and $c$ random insertions, respectively, and ${\cal{Y}}^d(\bfm x-a)$ and ${\cal{Y}}^i(\bfm x+c)$ denote the set of output sequences resulting from $a$ deletions from and $c$ random insertions into, the input sequence $\bfm x$, respectively. We denote the deletion pattern of length $d$ in a sequence of length $n$ with $K$ runs by $D(n;K;d)=(d_1,d_2,\dotsc  ,d_K)$, where $d_k$ denotes the number of deletions in the $k$-th run and $\sum_{k=1}^{K}d_k=d$. The outputs resulting from a given deletion pattern $D(n;K;d)=(d_1,d_2,\dotsc  ,d_K)$ (without any other error) are denoted by $D(n;K;d)*\bfm x(n;K)=(n_1-d_1,n_2-d_2,\dotsc  ,n_K-d_K)$. The set ${\cal{D}}{_K^n}(d)$ represents the set of all deletion patterns of length $d$ of a sequence of length $n$ and with $K$ runs.

\subsection{Organization of the Paper}
 In Section~\ref{main_approach}, we introduce our general approach for lower bounding the mutual information of the input and output sequences for insertion/deletion channels. In Section~\ref{deletion_channel}, we apply the introduced approach to the deletion-substitution and deletion-AWGN channels and present analytical lower bounds on their capacities, and compare the resulting expressions with earlier results. In Section~\ref{insertion_channel}, we provide lower bounds on the capacity of the random insertion channels and comment on our results with respect to the existing literature. In Section~\ref{numerical_ex}, we compute the lower bounds for a number of insertion/deletion channels, and finally, we provide our conclusions in Section~\ref{conclusion}. 

\section{Main Approach}\label{main_approach}
We rely on lower bounding the information capacity of memoryless channels with insertion or deletion errors directly as justified by~\cite{dobrushin}, where it is shown that, for a memoryless channel with synchronization errors, the Shannon's theorem on transmission rates applies and the information and transmission capacities are equal, and thus every lower bound on the information capacity of an insertion/deletion channel is a lower bound on the transmission capacity of the channel. Our approach is different than most existing work on finding lower bounds on the capacity of the insertion/deletion channels where typically the transmission capacity is lower bounded using a certain codebook and particular decoding algorithms. The idea we employ is similar to the work in~\cite{drinea2007} which also considers the information capacity $\displaystyle \lim_{N\to\infty}\frac{1}{N}\displaystyle\max_{P(\bfm X)}I(\bfm X;\bfm Y)$ and directly lower bounds it using a particular input distribution to arrive at an achievable rate result.

Our primary focus is on the small deletion and insertion probabilities. As also noted in~\cite{kanoria}, for such probabilities it is natural to consider binary i.u.d. input distribution. This is justified by noting that when $p_d=p_i=0$, i.e., for a binary symmetric channel, the capacity is achieved with independent and symmetric binary inputs, and hence we expect that for small insertion/deletion probabilities, binary i.u.d. inputs are not far from the optimal input distribution.

Our methodology is to consider a finite length transmission of i.u.d. bits over the insertion/deletion channel, and to compute (more precisely, lower bound) the mutual information between the input and the resulting output sequences. As proved in~\cite{dario} for a channel with deletion errors, such a finite length transmission in fact results in an upper bound on the mutual information supported by the insertion/deletion channels; however, as also shown in~\cite{dario}, if a suitable term is subtracted from the mutual information, a provable lower bound on the achievable rate, hence the channel capacity, results. The following theorem provides this result in a slightly generalized form compared to~\cite{dario}.

 \begin{thm}\label{lem_LB}
For binary input channels with i.i.d. insertion or deletion errors, for any input distribution and any $n>0$, the channel capacity $C$ can be lower bounded by
\begin{equation}\label{LB3}
C\geq \frac{1}{n}I(\bfm X;\bfm Y)-\frac{1}{n}H(\bfm T),
\end{equation}
where $$H(\bfm T)=-\sum_{j=0}^n \left[{n\choose
j}p^j(1-p)^{n-j}\log\left({n\choose j}p^j(1-p)^{n-j}\right)\right]$$ with
the understanding that $p=p_d$ for the deletion channel case and $p=p_i$ in the insertion
channel case, and $n$ is the length of the input sequence $\bfm X$.
\end{thm}
\begin{IEEEproof}
This is a slight generalization of a result in~\cite{dario} which shows that
Eq.~\eqref{LB3} is valid for the i.i.d. deletion channel. It is easy to see that~\cite{dario}, for any random process $\bfm T^N$, and for any input distribution $P(\bfm X^N)$, we have
\begin{equation}\label{LB1}
C\geq \lim_{N\to\infty}\frac{1}{N}I(\bfm X^N;\bfm Y^N,\bfm T^N)-\lim_{N\to\infty}\frac{1}{N}H(\bfm T^N),
\end{equation}
where $C$ is the capacity of the channel, $N$ is the length of the input sequence $\bfm X^N$ and $N=Qn$, i.e., the input bits in both insertion and deletion channels are divided into $Q$ blocks of length $n$ ($\bfm X^N=\{\bfm X_j\}_{j=1}^Q$). We define the random process $\bfm T^N$ in the following manner. For an i.i.d. insertion channel, $\bfm T^{N,i}$ is formed as the sequence $\bfm{T}^{N,i}=\{T^{i}_j\}_{j=1}^Q$ which denotes the number of insertions that occur in transmission of each block of length $n$. For a deletion channel, $\bfm{T}^{N,d}=\{T^{d}_j\}_{j=1}^Q$ represents the number of deletions occurring in transmission of each block. Since insertions (deletions) for different blocks are independent, the random variables $T_j=T^{i}_j$ ($T^{d}_j$) for $j\in\{1,\dotsc  ,Q\}$ are i.i.d., and transmission of different blocks are independent. Therefore, we can rewrite Eq.~\eqref{LB1} as
\begin{eqnarray}\label{LB2}
C&\geq& \frac{1}{n}I(\bfm X_j;\bfm Y_j)-\frac{1}{n}H(\bfm T_j)\nonumber\\
&=&\frac{1}{n}I(\bfm X;\bfm Y)-\frac{1}{n}H(\bfm T).
\end{eqnarray}
\noindent Noting that the random variable denoting the number of deletions or insertions as a result of $n$ bit transmission is binomial with parameters $n$ and $p_d$ (or, $p_i$) the result follows.
\end{IEEEproof}

\noindent Several comments on the specific calculations involved are in order. Theorem~\ref{lem_LB} shows that for any input
distribution and any transmission length, Eq.~\eqref{LB3} results in a
lower bound on the capacity of the channel with deletion or insertion errors. Therefore, employing any lower bound on the mutual information rate $\frac{1}{n}I(\bfm X;\bfm Y)$ in Eq.~\eqref{LB3} also results in a lower bound on the capacity of the insertion/deletion channel. Due to the fact that obtaining the exact value of the mutual information rate for any $n$ is infeasible, we first derive a lower bound on the mutual information rate for i.u.d. input sequences and then employ it in Eq.~\eqref{LB3}. Based on the
formulation of the mutual information, obviously
\begin{equation}\label{I}
I(\bfm X;\bfm Y)=H(\bfm Y)-H(\bfm Y|\bfm X),
\end{equation}
thus by calculating the exact value of the output entropy or lower bounding it and obtaining the exact value of the conditional output entropy or upper bounding it, the mutual information is lower bounded. For the models adopted in this paper, we are able to obtain the exact value of the output sequence probability distribution when i.u.d. input sequences are used, hence the exact value of the output entropy (the differential output entropy for the deletion-AWGN channel) is available. 

In deriving the conditional output entropies (the conditional differential entropy of the output sequence for the deletion-AWGN channel), we cannot obtain the exact probability of all the possible output sequences conditioned on a given input sequence. For deletion channels, we compute the probability of all possible deletion patterns for a given input sequence, and treat the resulting sequences as if they are all distinct to find a provable upper bound on the conditional entropy term. Clearly, we are losing some tightness, as different deletion patterns may result in the same sequence at the channel output. For the random insertion channel, we calculate the conditional probability of the output sequences resulting from at most one insertion, and derive an upper bound on the part of the conditional output entropy expression that results from the output sequences with multiple insertions.

\section{Lower Bounds on the Capacity of Noisy Deletion Channels}\label{deletion_channel}
As mentioned earlier, we consider two different variations of the binary deletion channel: i.i.d. deletion and substitution channel (deletion-substitution channel), and i.i.d. deletion channel in the presence of AWGN (deletion-AWGN channel). The results utilize the idea and approach of the previous section. We first give the results for the deletion-substitution channel, then for the deletion-AWGN channel. We note that the presented lower bounds can be also employed on the deletion-only channel if $p_e=0$ (or $\sigma^2=0$ for the deletion-AWGN channel). 

\subsection{Deletion-Substitution Channel}
In this section, we consider a binary deletion channel with substitution errors in which each bit is independently deleted with probability $p_d$, and transmitted bits are independently flipped with probability $p_e$. The receiver and the transmitter do not have any information about the position of deletions or the substitution errors. As shown in Fig.~\ref{fig:del_sub}, this channel can be considered as a cascade of an i.i.d. deletion channel with a deletion probability $p_d$ and output sequence $\bfm Y$, and a BSC with a cross-over error probability $p_e$ and output sequence $\bfm Y'$. For such a channel model the following lemma is a lower bound on the capacity.
\begin{figure}
    \centering
    \includegraphics[width=.43\textwidth]{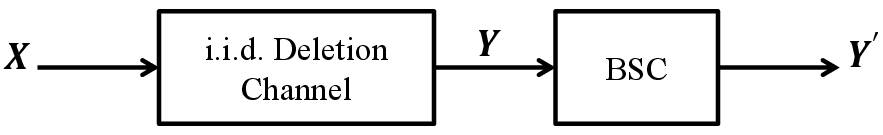}
    \caption{Deletion-substitution channel as a cascade of an i.i.d. deletion channel and a BSC.}
    \label{fig:del_sub}
\end{figure}

\begin{lem}\label{lemma_C_del_sub}
For any $n>0$, the capacity of the i.i.d. deletion-substitution channel $C_{ds}$, with a substitution probability $p_e$ and a deletion probability $p_d$, is lower bounded by
\begin{eqnarray}\label{eq_lb_del_sub}
C_{ds}&\geq& 1-p_d-H_b(p_d)-(1-p_d)H_b(p_e)\nonumber\\
&&+\frac{1}{n}\sum_{j=1}^n W_j(n){n\choose j}p_d^j(1-p_d)^{n-j},
\end{eqnarray}
where 
\begin{align}\label{eq_LB_deletion}
W_j(n)=&\sum_{l=1}^{n-1}2^{-l-1}(n-l+3)\sum_{j'=1}^j\frac{{{l}\choose j'}{{n-l}\choose {j-j'}}}{{n\choose j}}  \log{l\choose j'}\nonumber\\
&+2^{-n+1}\log{n\choose j},
\end{align}
and $H_b(p_d)=-p_d\log(p_d)-(1-p_d)\log(1-p_d)$.\hfill$\Box$
\end{lem}

Before proving the lemma, we would like to emphasize that the only existing analytical lower bound on the capacity of deletion-substitution channels  is derived in~\cite{gallager} (Eq.~\eqref{LB_gallager_del_sub}). In comparing the lower bound in Eq.~\eqref{LB_gallager_del_sub} with the lower bound in Eq.~\eqref{eq_lb_del_sub}, we observe that the new lower bound improves the previous one by \mbox{$\frac{1}{n}\sum_{j=1}^n W_j(n){n\choose j}p_d^j(1-p_d)^{n-j}-p_d$}, which is guaranteed to be positive.

A simplified form of the lower bound for small values of deletion probability can also be presented. By invoking the inequalities $(1-p)^m\geq[1-mp+{m\choose2}p^2-{m\choose3}p^3]$ and \mbox{$(1-p)^m\geq1-mp$}, and ignoring some positive terms ($p_d^j(1-p_d)^{n-j}$ for $j\geq3$), we can write
\begin{align}\label{C_d_2}
C_d\hspace*{-.01in}\geq& 1\hspace*{-.02in}-\hspace*{-.02in}H_b(p_d)\hspace*{-.02in}+\hspace*{-.02in}p_d(W_1(n)\hspace*{-.02in}-\hspace*{-.02in}1)\hspace*{-.02in}+\hspace*{-.02in}p_d^2\frac{n\hspace*{-.02in}-\hspace*{-.02in}1}{2}\left(W_2(n)\hspace*{-.02in}-\hspace*{-.02in}2W_1(n)\right)\nonumber\\
&\ +p_d^3{{n-1}\choose2}\left(W_1(n)-W_2(n)\right)-p_d^4{{n-1}\choose3}W_1(n).\nonumber
\end{align}

By utilizing $p_e=0$ in Eq.~\eqref{eq_lb_del_sub}, we can obtain a lower bound on the capacity of the deletion-only channel as given in the following corollary.
\begin{cor} \label{cor_C_deletion}
For any $n>0$, the capacity of an i.i.d. deletion channel $C_d$, with a deletion probability of $p_d$ is lower bounded by
\vspace*{-.1in}\begin{equation}\label{eq_LB_deletion_cor}
C_d\geq 1-p_d-H_b(p_d)+\frac{1}{n}\sum_{j=1}^n W_j(n){n\choose j}
p_d^j(1-p_d)^{n-j}.
\end{equation}
\end{cor}
\noindent We also would like to make a few comments on the result of the Corollary~\ref{cor_C_deletion}. First of all, the lower bound~\eqref{eq_LB_deletion_cor} is tighter than the one proved in~\cite{gallager} (Eq.~\eqref{eq:LB-gallager} with $p_i=p_e=0$) which is the simplest analytical lower bound on the capacity of the deletion channel. The amount of improvement in~\eqref{eq_LB_deletion_cor} over the one in~\eqref{eq:LB-gallager} is $\frac{1}{n}\sum_{j=1}^n W_j(n){n\choose j}p_d^j(1-p_d)^{n-j}-p_d$, which is guaranteed to be  positive.

In~\cite{kanoria}, it is shown that
\begin{equation}\label{C_expansion}
C_d=1+p_d\log(p_d)-A_1p_d+O(p_d^{1.4}),
\end{equation}
where $A_1=\log(2e)-\sum_{l=1}^{\infty} 2^{-l-1}l\log(l)$. A similar result in~\cite{asymptotic} is provided, that is $C_d \leq 1-(1-O(p_d))H_b(p_d)$, which shows that $1-H_b(p_d)$ is a tight lower bound for small deletion probabilities. If we consider the new capacity lower bound in~\eqref{eq_LB_deletion_cor}, and represent $(1-p_d)\log(1-p_d)$ by its Taylor series expansion, we can readily write
\begin{equation}
C_d\geq 1+p_d\log(p_d)-\left(\log(2e)-W_1(n)\right)p_d+p_d^2f(n,p_d),\nonumber
\end{equation}
where $f(n,p_d)$ is a polynomial function. On the other hand for $W_1(n)$, if we let $n$ go to infinity, we have
\begin{align}
\lim_{n\rightarrow\infty}\hspace*{-.05in} W_1(n)=&\lim_{n\rightarrow\infty}\bigg[\frac{1}{n}\sum_{l=1}^{n-1}2^{-l-1}(n\hspace*{-.02in}-\hspace*{-.02in}l\hspace*{-.02in}+\hspace*{-.02in}3)l\log(l)+\frac{\log(n)}{2^{n-1}}\bigg]\nonumber\\
=&\sum_{l=1}^{\infty}2^{-l-1}l\log(l).
\end{align}
\noindent Therefore, we observe that the lower bound~\eqref{eq_LB_deletion_cor} captures the first order term of the capacity expansion~\eqref{C_expansion}. This is an important result as the capacity expansions in~\cite{kanoria,asymptotic} are asymptotic and do not lend themselves for a numerical calculation of the transmission rates for any non-zero value of the deletion probability.

We need the following two propositions in the proof of Lemma~\ref{lemma_C_del_sub}. In Proposition~\ref{lemma_H(Y)_deletion_substitution}, we obtain the exact value of the output entropy in the deletion-substitution channel with i.u.d. input sequences, while Proposition~\ref{lemma_Hyx_ub_deletion_substitution} gives an upper bound on the conditional output entropy with i.u.d. bits transmitted through the deletion-substitution channel.

\begin{prop}\label{lemma_H(Y)_deletion_substitution}
For an i.i.d. deletion-substitution channel with i.u.d. input sequences of length $n$, we have
\begin{eqnarray}\label{eq_HY_deletion_substitution}
  H(\bfm Y')= n(1-p_d)+H(\bfm T),
\end{eqnarray}
where $\bfm Y'$ denotes the output sequence of the deletion-substitution channel and $H(\bfm T)$ is as defined in Eq.~\eqref{LB3}.
\end{prop}

\begin{IEEEproof}
By using the facts that all the elements of the set ${\cal Y}_{-j}^d$ are identically distributed, which are inputs into the BSC channel, and a fixed length i.u.d. input sequence into a BSC result in i.u.d. output sequences,
all elements of the set ${\cal Y'}_{-j}^d$ are also identically distributed. Hence,
\begin{equation}\label{Q_y_deletion_substitution}
    P(\bfm y'(n-j))=\frac{1}{2^{n-j}}{n\choose{j}}p_d^j(1-p_d)^{n-j},
\end{equation}
where $\displaystyle{n\choose{j}}p_d^j(1-p_d)^{n-j}$ is the probability of exactly $j$ deletions occurring in $n$ use of the channel. Therefore, we obtain
\begin{align}
H(\bfm Y')&=\sum_{\bfm y'}-P(\bfm y')\log(P(\bfm y'))\nonumber\\
&=\sum_{j=0}^{n}{n\choose{j}}p_d^j(1-p_d)^{n-j}\log\left(\frac{2^{n-j}}{{n\choose{j}}p_d^j(1-p_d)^{n-j}}\right)\nonumber\\
&=n(1-p_d)+H(\bfm T),
\end{align}
which concludes the proof.
\end{IEEEproof}

\begin{prop}\label{lemma_Hyx_ub_deletion_substitution}
For a deletion-substitution channel with i.u.d. input sequences, the entropy of the output $\bfm Y'$ conditioned on the input $\bfm X$ of length $n$ bits, is upper bounded by
\begin{eqnarray}\label{Hyx_UB_deletion_substitution}
H(\bfm Y'|\bfm X) &\leq& nH_b(p_d)-\sum_{j=1}^n W_j(n){n\choose
j}p_d^j(1-p_d)^{n-j}\nonumber\\
&&+n(1-p_d)H_b(p_e),
\end{eqnarray}
where $W_j(n)$ is given in Eq.~\eqref{eq_LB_deletion}.
\end{prop}

\begin{IEEEproof}
To obtain the conditional output entropy, we need to compute the probability of all possible output sequences resulting from every possible input sequence $\bfm x$, i.e., $P(\bfm Y'|\bfm x)$. For a given $\bfm x=(b;n_1,n_2,\dotsc  ,n_k)$ and for a specific deletion pattern $D(n;K;j)=(j_1,\dotsc  ,j_K)$ in which $j_k$ denotes the number of deletions in the $k$-th run, we can write
\begin{align}\label{Q(y|x(n,K))_del_sub}
    P\bigg(D(n;K;j)=&(j_1,\dotsc  ,j_K)\bigg|\bfm x(b;n_1,\dotsc  ,n_K)\bigg)\nonumber\\
    &= {{n_1}\choose j_1}\dotsc  {{n_K}\choose j_K} p_d^j(1-p_d)^{n-j}.
\end{align}
Furthermore, for every $D(n;K;j)$, we can write
\begin{equation}\label{eq_bsc}
P\left(\hspace*{-.02in}\bfm y'\bigg|D*\bfm x(n;K)\hspace*{-.02in}\right)\hspace*{-.03in}=\hspace*{-.03in}\left\{\begin{array}{ccc}
     \hspace*{-.05in} p_e^{s}(1-p_e)^{n-j-s} & \mbox{ if } & |\bfm y'|=n-j, \\
      0 && \mbox{otherwise,}
    \end{array}\right.
\end{equation}
where $s=d_H\left(\bfm y';D(n;K;j)*\bfm x(n;K)\right)$, and $d_H=(\bfm a;\bfm b)$ is the Hamming distance between two sequences $\bfm a$ and $\bfm b$.
On the other hand, for every output sequence of length $n-j$, conditioned on a given input $\bfm x(n;K)$, we have
\small\begin{align}
P\hspace*{-.02in}\bigg(\hspace*{-.04in}\bfm y'(n\hspace*{-.02in}-\hspace*{-.02in}j)\hspace*{-.015in}\bigg|\hspace*{-.015in}\bfm x(n;K)\hspace*{-.05in}\bigg)\hspace*{-.03in}=\hspace*{-.15in}\sum_{D\in {\cal{D}}{_K^n}(j)}\hspace*{-.15in} P\hspace*{-.03in}\left(\hspace*{-.04in}\bfm y'(n\hspace*{-.02in}-\hspace*{-.02in}j)\hspace*{-.015in}\bigg|\hspace*{-.015in}D,\bfm x(n;K)\hspace*{-.05in}\right)\hspace*{-.03in}P\hspace*{-.03in}\left(\hspace*{-.04in}D\bigg|\bfm x(n;K)\hspace*{-.05in}\right)\hspace*{-.02in}.\nonumber
\end{align}\normalsize
However, there is a difficulty as two different possible deletion patterns, $D(n;K;j)=(j_1,\cdots,j_K)$ and $D'(n;K;j)=(j'_1,\cdots,j'_K)$, under the same substitution error pattern, i.e., the substitution errors occur at the same positions on  $D(n;K;j)*x(n;K)$ and $D'(n;K;j)*x(n,K)$, may convert a given input sequence $\bfm x(n;K)$ into the same output sequence, i.e., $D(n;K;j)*x(n;K) =D'(n;K;j)*x(n,K)$. This occurs when successive runs are completely deleted, for example, in transmitting
$(1;2,1,2,3,2)=1101100011$, if the second, third and fourth runs are
completely deleted, by deleting one bit from the first run, $(1,1,2,3,0)*(1;2,1,2,3,2)=(1;1,0,0,0,2)=111$, or from the last run, $(0,1,2,3,1)*(1;2,1,2,3,2)=(1;2,0,0,0,1)=111$, the same
output sequences are obtained. This difficulty can be addressed using
\begin{equation}\label{convex}
\sum_{t} -p_t \left(\log\sum_{t'} p_{t'}\right)\leq \sum_{t} -p_t\log(p_t),
\end{equation}
which is trivially valid for any set of probabilities $(p_1,\dotsc  ,p_t,\dotsc  )$. Therefore, we can write
\begin{align}\label{eq_split}
&-P(\bfm y'|\bfm x)\log\left(P(\bfm y'|\bfm x)\right)\nonumber\\
&=\hspace*{-.025in}-\hspace*{-.16in}\sum_{D\in {\cal{D}}{_K^n}(j)}\hspace*{-.16in} P(\bfm y'|D\hspace*{-.02in}*\hspace*{-.02in}\bfm x)\hspace*{-.015in}P(D|\bfm x)\hspace*{-.02in}\log\hspace*{-.04in}\left(\hspace*{-.035in}\sum_{D'\in {\cal{D}}{_K^n}(j)}\hspace*{-.16in} P(\bfm y'|D'\hspace*{-.02in}*\hspace*{-.02in}\bfm x)P(D'|\bfm x)\hspace*{-.04in}\right)\nonumber\\
&\leq -\hspace*{-.1in}\sum_{D\in {\cal{D}}{_K^n}(j)}\hspace*{-.1in} P(\bfm y'|D*\bfm x)P(D|\bfm x)\log\bigg(P(\bfm y'|D*\bfm x)P(D|\bfm x)\bigg).
\end{align}
Hence, for a specific $\bfm x(b;n;K^x)=(b;n_1^x,\dotsc  ,n_{K^x}^x)$, we obtain (for more details see Appendix~\ref{app_hyx_del_sub})
\begin{align}
&H\bigg(\bfm Y'\bigg|\bfm x(b;n;K^x)\bigg)\leq n H_b (p_d)+n(1-p_d) H_b (p_e)\nonumber\\
&-\sum_{j=0}^n p_d^j(1-p_d)^{n-j}\sum_{k=1}^{K^x}\sum_{j_k=0}^j {{n_k^x}\choose j_k}{{n-n_k^x}\choose {j-j_K}} \log{{n_k^x}\choose j_k}.\nonumber
\end{align}
Therefore, by considering i.u.d. input sequences, we have
\begin{align}\label{Hhs}
H(\bfm Y'|\bfm X)=&\hspace*{-.02in}\sum_{\bfm x\in\cal
X}\frac{1}{2^n}H(\bfm Y'|\bfm x)\leq nH_b(p_d)\hspace*{-.02in}+\hspace*{-.02in}n(1\hspace*{-.02in}-\hspace*{-.02in}p_d) H_b(p_e)\nonumber\\
&\hspace*{-.75in}-\hspace*{-.02in}\sum_{j=0}^n
\frac{p_d^j(1\hspace*{-.02in}-\hspace*{-.02in}p_d)^{n-j}}{2^n}\hspace*{-.02in}\sum_{\bfm x\in\cal
X}\sum_{k=1}^{K^x}\sum_{j_k=0}^j\hspace*{-.02in} {{n_k^x}\choose
j_k}\hspace*{-.02in}{{n\hspace*{-.02in}-\hspace*{-.02in}n_k^x}\choose {j\hspace*{-.02in}-\hspace*{-.02in}j_k}}\hspace*{-.02in} \log{{n_k^x}\choose j_k}.
\end{align}

On the other hand, we can write
\begin{align}\label{havijuriPR}
\sum_{\bfm x\in\cal X}\frac{1}{2^n}&\sum_{k=1}^{K^x}\sum_{j_k=0}^j {{n_k^x}\choose j_k}{{n-n_k^x}\choose {j-j_k}} \log{{n_k^x}\choose j_k}\nonumber\\
&=\sum_{j'=0}^j\sum_{l=1}^n P_R(l,n){{l}\choose j'}{{n-l}\choose {j-j'}} \log{l\choose j'},
\end{align}
where $P_R(l,n)$ denotes the probability of having a run of length $l$ in an input sequence of length $n$. It is obvious that $P_R(n,n)=\frac{2}{2^n}$. Due to the fact that, for $1\leq l\leq n-1$, there are $\displaystyle{{n\hspace*{-.01in}-\hspace*{-.01in}l\hspace*{-.01in}-\hspace*{-.01in}1}\choose {K\hspace*{-.01in}-\hspace*{-.01in}2}}$ possibilities to have a run of length $l$ in a sequence with $K$ runs, we can write
\begin{equation}\label{P_R}
    P_R(l,n)\hspace*{-.01in}=\hspace*{-.01in} \frac{2}{2^n}\sum_{K=2}^{n-l+1}{{n\hspace*{-.01in}-\hspace*{-.01in}l\hspace*{-.01in}-\hspace*{-.01in}1}\choose {K-2}}K=2^{-l-1}(n\hspace*{-.01in}-\hspace*{-.01in}l\hspace*{-.01in}+\hspace*{-.01in}3).
\end{equation}
Finally, by substituting Eqs.~\eqref{havijuriPR} and~\eqref{P_R} in Eq.~\eqref{Hhs}, Eq.~\eqref{Hyx_UB_deletion_substitution} results, completing the proof. 
\end{IEEEproof}

We can now complete the proof of the main lemma of the section.

\textit{\textbf{Proof of Lemma~\ref{lemma_C_del_sub}}}: In Theorem~\ref{lem_LB}, we showed that for any input distribution and any transmission length, Eq.~\eqref{LB3} results in a lower bound on the capacity of the channel with i.i.d. deletion errors. On the other hand, any lower bound on the information rate can also be used to derive a lower bound on the capacity. Due to the definition of the mutual information, Eq.~\eqref{I}, by obtaining the exact value of the output entropy (Proposition~\ref{lemma_H(Y)_deletion_substitution}) and upper bounding the conditional output entropy (Proposition~\ref{lemma_Hyx_ub_deletion_substitution}) the mutual information is lower bounded. Finally, by substituting Eqs.~\eqref{eq_HY_deletion_substitution} and~\eqref{Hyx_UB_deletion_substitution} into Eq.~\eqref{LB3}, Lemma~\ref{lemma_C_del_sub} is proved.\hfill$\Box$ 

At this point we digress to point out that the result in the above lemma can also be obtained using a simpler approach as pointed out by one of the reviewers (details are given in Appendix A). That is,  a lower bound on the deletion-substitution channel capacity can be provided in terms of the deletion-only channel capacity as (this is also a special case of a result in~\cite{ISIT-noisy})
 
\begin{equation}\label{eq-Cds-Cd}
C_{ds}\geq C_d-(1-p_d)H_b(p_e).
\end{equation}
Therefore, computing the mutual information rate of the deletion-only channel for i.u.d. input sequences and substituting it in the above inequality results in a lower bound on $C_{ds}$. It can be verified that the same procedure as in the proof of Lemma~\ref{lemma_C_del_sub} gives
\begin{equation}
C_d\geq 1-p_d-H_b(p_d)+\frac{1}{n}\sum_{j=1}^n W_j(n){n\choose j}
p_d^j(1-p_d)^{n-j},\nonumber
\end{equation}
and substituting this into Eq.~\eqref{eq-Cds-Cd} concludes the proof of Lemma~\ref{lemma_C_del_sub}. 

\subsection{Deletion-AWGN Channel}
In this section, a binary deletion channel in the presence of AWGN is considered, where the bits are transmitted using binary phase shift keying (BPSK) and the received signal contains AWGN in addition to the deletion errors. As illustrated in Fig.~\ref{fig:del_AWGN},
\begin{figure}
    \centering
    \includegraphics[width=.43\textwidth]{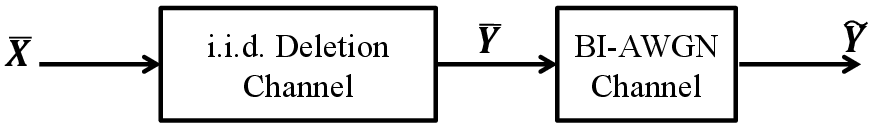}
    \caption{Deletion-AWGN channel as a cascade of an i.i.d. deletion channel and a BI-AWGN channel.}
    \label{fig:del_AWGN}
\end{figure}
this channel can be considered as a cascade of two
independent channels where the first channel is an i.i.d. deletion channel and the second one is a BI-AWGN channel. We use $\xbar{\bfm X}$ to denote the input sequence to the first channel which is a BPSK modulated version of the binary input sequence $\bfm X$, i.e., $\bar{x}_i=1-2x_i$, and $\xbar{\bfm Y}$ to denote the output sequence of the first channel input to the second one. $\widetilde{\bfm Y}$ is the output sequence of the second channel that is the noisy version of $\xbar{\bfm Y}$, i.e., $\widetilde{y}_i= \bar{y}_i + z_i$, in which $z_i$'s are i.i.d. Gaussian random variables with zero mean and a variance of $\sigma^2$, and ${\widetilde{y}}_i$ and
$\bar{y}_i$ are the $i^{th}$ received and transmitted bits of the second channel, respectively. Therefore, for the probability density function of the $i^{th}$ channel output, we have
\begin{align}\label{P(tilday)}
f_{\widetilde{y}_i}(\eta)=&f_{\widetilde{y}_i}(\eta|\bar{y}_i=1)P(\bar{y}_i\hspace*{-.02in}=\hspace*{-.02in}1)\hspace*{-.03in}+\hspace*{-.03in}f_{\widetilde{y}_i}(\eta|\bar{y}_i=-1)P(\bar{y}_i\hspace*{-.02in}=\hspace*{-.02in}-1)\nonumber\\
=&\frac{1}{\sqrt{2\pi}\sigma}\left[P(\bar{y}_i=1)e^{-\frac{(\eta-1)^2}{2\sigma^2}}+P(\bar{y}_i=-1)e^{-\frac{(\eta+1)^2}{2\sigma^2}}\right]\hspace*{-.02in}.
\end{align}
In the following lemma, an achievable rate is provided over this channel.
\begin{lem}\label{lemma_C_del_AWGN}
For any $n>0$, the capacity of the deletion-AWGN channel with a deletion
probability of $p_d$ and a noise variance of $\sigma^2$ is lower bounded by
\begin{align}\label{LB_del_AWGN}
\hspace*{-.083in} C_{d,AWGN}\geq & 1-p_d+\frac{1}{n}\sum_{j=1}^n W_j(n){n\choose
j}p_d^j(1-p_d)^{n-j}\nonumber\\
&-H_b(p_d)-(1-p_d)E\left[\log\left(1+e^{\frac{-2\bfm
z}{\sigma^2}}\right)\right],
\end{align}
where $W_j(n)$ is as given in Eq.~\eqref{eq_LB_deletion}, $E[.]$ is statistical expectation, and $\bfm z\sim{\cal{N}}(0,\sigma^2)$.\hfill $\Box$
\end{lem}

Before giving the proof of the above lemma, we provide several comments about the result. First, the desired lower bound in Eq.~\eqref{LB_del_AWGN} is the only analytical lower bound on the capacity of the deletion-AWGN channel. In the current literature, there are only simulation based lower bounds, e.g.~\cite{junhu}, which employs Monte-Carlo simulation techniques. Furthermore, the procedure employed in~\cite{junhu} is only useful for deriving lower bounds for small values of deletion probability, e.g., $p_d\leq0.1$, while the lower bound in Eq.~\eqref{LB_del_AWGN} is useful for a much wider range.

For $p_d=0$, the lower bound in Eq.~\eqref{LB_del_AWGN} is equal to ${1-E\left[\log(1+e^{\frac{-2\bfm
z}{\sigma^2}})\right]}$ which is the capacity of the BI-AWGN channel~\cite[p.~362]{proakis}. Finally, we note that the term in Eq.~\eqref{LB_del_AWGN} which contains $E\left[\log(1+e^{\frac{-2\bfm
z}{\sigma^2}})\right]$ can be easily computed by numerical integration with an arbitrary accuracy (it involves only an one-dimensional integral).

We need the following two propositions in the proof of Lemma~\ref{lemma_C_del_AWGN}. In the following proposition, the exact value of the differential output entropy in the deletion-AWGN channel with i.u.d. input bits is calculated.

\begin{prop}\label{lemma_H(Y)_del_AWGN}
For an i.i.d. deletion-AWGN channel with i.u.d. input sequences of length $n$, we have
\begin{align}\label{eq_HY_del_AWGN}
  h(\widetilde{\bfm Y})=& n(1-p_d)\left(\log\left(2\sigma\sqrt{2\pi e }\right)-E\left[\log\left(1+e^{-\frac{2\bfm z}{\sigma^2}}\right)\right]\right)\nonumber\\
&  +H(\bfm T),
\end{align}
where $h(.)$ denotes the differential entropy function, $\widetilde{\bfm Y}$ denotes the output of the deletion-AWGN channel, \mbox{$\bfm z\sim{\cal{N}}(0,\sigma^2)$}, and $H(\bfm T)$ is as defined in Eq.~\eqref{LB3}.
\end{prop}

\begin{proof}
For the differential entropy of the output sequence, we can write
\begin{eqnarray}\label{eq_HYj}
h(\widetilde{\bfm Y})&=&h(\widetilde{\bfm Y})+H(\bfm T|\widetilde{\bfm Y})\nonumber\\
&=&h(\widetilde{\bfm Y},\bfm T)\nonumber\\
&=&h(\widetilde{\bfm Y}|\bfm T)+H(\bfm T),
\end{eqnarray}
where the first equality results by using the fact that by knowing the received sequence, the number of deletions is known and $\bfm T$ is determined, i.e., $H(\bfm T|\widetilde{\bfm Y})=0$, and the last equality is obtained by using a different expansion of $h(\widetilde{\bfm Y},\bfm T)$. On the other hand, we can write
\begin{eqnarray}\label{eq_HYT}
h(\widetilde{\bfm Y}|\bfm T)&=&\sum_{j=0}^{n}h(\widetilde{\bfm Y}|\bfm T=j)P(\bfm T=j)\nonumber\\
&=&\sum_{j=0}^{n}h(\widetilde{\bfm Y}|\bfm T=j){n\choose j}p_d^j(1-p_d)^{n-j}.
\end{eqnarray}
Due to the fact that all the elements of the set ${\xbar{\cal Y}}_{-j}^d$ are i.i.d., we have $P(\bar{\bfm y}(n-j))=P(\bar{\bfm y},\bfm T=j)=\frac{1}{2^{n-j}}{n\choose{j}}p_d^j(1-p_d)^{n-j}$. Therefore, we can write
\begin{eqnarray}\label{P(ybar)}
P(\bar{\bfm y}|\bfm T=j)&=&\frac{P(\bar{\bfm y},\bfm T=j)}{P(\bfm T=j)}=\frac{1}{2^{n-j}},
\end{eqnarray}
and as a result $P(\bar{y}_i=1|\bfm T=j)=P(\bar{y}_i=-1|\bfm T=j)=\frac{1}{2}$ (for $1\leq i\leq n-j$). By employing this result in Eq.~\eqref{P(tilday)}, we have
\begin{equation}\label{P(ytilda)}
f_{\widetilde{y}_i}(\eta)=\frac{1}{2 \sqrt{2\pi}\sigma}\left[e^{-\frac{(\eta-1)^2}{2\sigma^2}}+e^{-\frac{(\eta+1)^2}{2\sigma^2}}\right],
\end{equation}
where $f_{\widetilde{y}_i}(\eta)$ denotes the probability density function (PDF) of the continuous random variable $\widetilde{y}_i$. Noting also that the deletions happen independently and $\widetilde{y}_i$'s are i.i.d., we can write
\begin{align}
h(\widetilde{\bfm Y}|\bfm T=j)=&(n-j)h(\widetilde{y}_i)\nonumber\\
=&(n\hspace*{-.01in}-\hspace*{-.01in}j)\int_{-\infty}^\infty -f_{\widetilde{y}_i}(\eta)\log\left(f_{\widetilde{y}_i}(\eta)\right)d\eta\nonumber\\
=&(n\hspace*{-.01in}-\hspace*{-.01in}j)\hspace*{-.02in}\left(\log\hspace*{-.01in}\left(2\sigma\sqrt{2\pi e}\right)\hspace*{-.02in}-\hspace*{-.03in} E\hspace*{-.02in}\left[\log\left(1\hspace*{-.015in}+\hspace*{-.015in}e^{-\frac{2\bfm z}{\sigma^2}}\right)\right]\hspace*{-.015in}\right)\hspace*{-.015in}.\nonumber
\end{align}
By substituting the above equation into Eq.~\eqref{eq_HYT}, we obtain
\begin{align}\label{eq_HYT2}
h(\widetilde{\bfm Y}|\bfm T)=&\sum_{j=0}^{n}(n-j){n\choose j}p_d^j(1-p_d)^{n-j}\times \nonumber\\ 
&\times\left(\log(2\sigma\sqrt{2\pi e})-E\left[\log(1+e^{-\frac{2\bfm z}{\sigma^2}})\right]\right)\nonumber\\
&\hspace*{-.55in}= n(1-p_d)\left(\log(2\sigma\sqrt{2\pi e})-E\left[\log(1+e^{-\frac{2\bfm z}{\sigma^2}})\right]\right),
\end{align}
and by using Eqs.~\eqref{eq_HYT2} and~\eqref{eq_HYj}, Eq.~\eqref{eq_HY_del_AWGN} is obtained.
\end{proof}

In the following proposition, we derive an upper bound on the differential entropy of the output conditioned on the input for deletion-AWGN channel.

\begin{prop}\label{lemma_Hyx_UB_del_AWGN}
For a deletion-AWGN channel with i.u.d. input bits, the differential entropy of the output sequence $\widetilde{\bfm Y}$ conditioned on the input $\bfm X$ of length $n$, is upper bounded by
\begin{eqnarray}\label{Hyx_UB_del_AWGN}
h(\widetilde{\bfm Y}|\bfm X) \leq & nH_b(p_d)-\sum_{j=1}^n W_j(n){n\choose
j}p_d^j(1-p_d)^{n-j}\nonumber\\
&+n(1-p_d)\log(2\sigma\sqrt{2\pi e }),
\end{eqnarray}
where $W_j(n)$ is given in Eq.~\eqref{eq_LB_deletion}.
\end{prop}
\begin{IEEEproof}
For the conditional differential entropy of the output sequence given the length $n$ input $\bfm X$, we can write
\begin{eqnarray}\label{eq_H(Yj|Xj)}
h(\widetilde{\bfm Y}|\bfm X)&=&h(\widetilde{\bfm Y}|\bfm X)+H(\bfm T|\widetilde{\bfm Y},\bfm X)\nonumber\\
&=&H(\bfm T)+h(\widetilde{\bfm Y}|\bfm T,\bfm X),
\end{eqnarray}
where the first equality follows since by knowing $\bfm X$ and $\widetilde{\bfm Y}$, the number of deletions is known, i.e., $H(\bfm T|\widetilde{\bfm Y},\bfm X)=0$. The second equality is obtained by using a different expansion of $h(\widetilde{\bfm Y},\bfm T|\bfm X)$ and also using the fact that the deletion process is independent of the input $\bfm X$, i.e., $H(\bfm T|\bfm X)=H(\bfm T)$. Furthermore, we have
\vspace*{-.05in}\begin{eqnarray}
h(\widetilde{\bfm Y}|\bfm T,\bfm X)&=&\sum_{j=0}^nh(\widetilde{\bfm Y}|\bfm X,\bfm T=j)P(\bfm T=j)\nonumber\\
&=&\sum_{j=0}^nh(\widetilde{\bfm Y}|\bfm X,\bfm T=j) {n\choose j}p_d^j(1-p_d)^{n-j}.\nonumber
\end{eqnarray}
To obtain $h(\widetilde{\bfm Y}|\bfm X,\bfm T=j)$, we need to compute $f_{\widetilde{\bfm y}|\bfm x,j}(\eta)$ for any given input sequence $\bfm x=(b;n_1,n_2,\dotsc  ,n_K )$ and different values of $j$.
As in the proofs of Proposition~\ref{lemma_Hyx_ub_deletion_substitution}, if we consider the outputs of the deletion channel resulting from different deletion patterns of length $j$ from a given $\bfm x$, as if they are distinct and also use the result in Eq.~\eqref{convex}, an upper bound on the differential output entropy conditioned on the input sequence $\bfm X$ results. We relegate the details of this computation and completion of the proposition proof to Appendix~\ref{app_H(YX)_ub_awgn}.
\end{IEEEproof}

We can now state the proof of the main lemma of the section.

\textit{\textbf{Proof of Lemma~\ref{lemma_C_del_AWGN}}}: By substituting the exact value of the differential output entropy in Eq.~\eqref{eq_HY_del_AWGN}, and the upper bound~\eqref{Hyx_UB_del_AWGN} on the differential output entropy conditioned on the input in Eq.~\eqref{I}, a lower bound on the mutual information rate of the deletion-AWGN channel is obtained, hence the lemma is proved.\hfill$\Box$

\section{Lower Bounds on the Capacity of Random Insertion Channels}\label{insertion_channel}
We now turn our attention to the random insertion channels and derive lower bounds on the capacity of random insertion channels by employing the approach proposed in Section~\ref{main_approach}. We consider the Gallager model~\cite{gallager} for insertion channels in which every transmitted bit is independently replaced by two random bits with probability of $p_i$ while neither the receiver nor the transmitter have any information about the position of the insertions. The following lemma provides the main result of this section.
\begin{lem}\label{lemma_C_insertion}
For any $n>0$, the capacity of the random insertion channel $C_i$, is lower bounded by
\begin{align}\label{eq_LB_insertion}
&C_i\geq (1\hspace*{-.01in}-\hspace*{-.01in}p_i)^n\hspace*{-.02in}-\hspace*{-.02in}H_b(p_i)
\hspace*{-.03in}+\hspace*{-.03in}\Bigg(\hspace*{-.03in}S(n)\hspace*{-.02in}-\hspace*{-.02in}\frac{3n+1}{4n}+n\hspace*{-.03in}\Bigg)p_i(1\hspace*{-.01in}-\hspace*{-.01in}p_i)^{n-1}
\nonumber\\
&+\hspace*{-.02in}\frac{\log\hspace*{-.02in}{n\choose 2}}{n}\hspace*{-.02in}\left(1\hspace*{-.02in}-\hspace*{-.02in}(1\hspace*{-.02in}-\hspace*{-.02in}p_i)^{n}\hspace*{-.02in}-\hspace*{-.04in}np_i(1\hspace*{-.02in}-\hspace*{-.02in}p_i)^{n-1}\hspace*{-.02in}-\hspace*{-.04in}p_i^n\hspace*{-.02in}-\hspace*{-.02in}np_i^{n-1}(1\hspace*{-.02in}-\hspace*{-.02in}p_i)\right)\nonumber\\
&+p_i^{n-1}(1-p_i)\log(n),
\end{align}
where 
\begin{align}
S(n)=\frac{1}{4n}\sum_{l=1}^{n-1}2^{-l}\bigg[&(n+1-l)(l+2)\log(l+2)\nonumber\\
&+2(l+1)\log(l+1)\bigg]+\frac{\log(n)}{2^{n+1}}.\nonumber \hspace*{.22in} \Box
\end{align}
\end{lem}

To the best of our knowledge, the only analytical lower bound on the capacity of the random insertion channel is derived in~\cite{gallager} (i.e., Eq.~\eqref{eq:LB-gallager} for $p_d=p_e=0$). Our result improves upon this result for small values of insertion probabilities as will be apparent with numerical examples.

Similar to the deletion-substitution channel case, 
we can write a simpler lower bound as
\begin{align}\label{eq_LB_insertion2}
C_i\geq & 1-H_b(p_i)+\bigg(S(n)-\frac{3n+1}{4n}\bigg)p_i\nonumber\\
&-\frac{n-1}{2}\left(2S(n)-\frac{3n+1}{2n}+n-\log{n\choose2}\right)p_i^2\nonumber\\
&-{{n-1}\choose2}\left(\log{n\choose2}-S(n)-\frac{2n}{3}+\frac{3n+1}{4n}\right)p_i^3\nonumber\\
&-{{n-1}\choose3}\left(S(n)+n-\frac{3n+1}{4n}\right)p_i^4.
\end{align}
For instance, for $n=10$, Eq.~\eqref{eq_LB_insertion2} evaluates to
\begin{align}\label{I_10}
C_i\geq&\ 1-H_b(p_i)+1.1591p_i-30.7184p_i^2+1.0502\times10^2p_i^3\nonumber\\
&\ -1.3391\times10^{3}p_i^4.
\end{align}

To prove the above lemma, we need the following two propositions. The output entropy of the random insertion channel with i.u.d. input sequences is calculated in the first one.

\begin{prop}\label{lemma_H(Y)_insertion}
For a random insertion channel with i.u.d. input sequences of length $n$, we have
\begin{eqnarray}\label{eq_HY_insertion}
  H(\bfm Y)= n(1+p_i)+H(\bfm T).
\end{eqnarray}
where $\bfm Y$ denotes the output sequence and $H(\bfm T)$ is as defined in Eq.~\eqref{LB3}.
\end{prop}

\begin{proof}
Similar to the proof of Proposition~\ref{lemma_H(Y)_deletion_substitution}, we use the fact that 
\vspace*{-.1in}\begin{equation}\label{Q_y_insertion}
    P(\bfm y(n+j))=\frac{1}{2^{n+j}}{n\choose{j}}p_i^j(1-p_i)^{n-j}.
\end{equation}
 Therefore, by employing Eq.~\eqref{Q_y_insertion} in computing the output entropy, we obtain
\vspace*{-.05in} \begin{align}
 H(\bfm Y)=& -\sum_{j=0}^n {n\choose{j}}p_i^j(1-p_i)^{n-j}\log\bigg(\frac{{n\choose{j}}p_i^j(1-p_i)^{n-j}}{2^{n+j}}\bigg)\nonumber\\
=&n(1+p_i)+H(\bfm T).
 \end{align}
\end{proof}

In the following proposition, we present an upper bound on the conditional output entropy of the random insertion channel with i.u.d. input sequences for a given input of length $n$.
\begin{prop}\label{lem_I(X;Y)_insertion}
For a random insertion channel with input and output sequences denoted by $\bfm X$ and $\bfm Y$, respectively, with i.u.d. input sequences of length $n$, we have
\begin{align}\label{eq_H(Y|X)_ins}
H(\bfm Y|\bfm X)\leq& n(1+p_i)+ n H_b(p_i) - n(1-p_i)^n \nonumber\\
&\hspace*{-.7in}-\left(1\hspace*{-.03in}-\hspace*{-.03in}(1\hspace*{-.03in}-\hspace*{-.03in}p_i)^{n}\hspace*{-.03in}-\hspace*{-.03in}np_i(1\hspace*{-.03in}-\hspace*{-.03in}p_i)^{n-1}\hspace*{-.03in}-\hspace*{-.03in}p_i^n\hspace*{-.03in}-\hspace*{-.03in}np_i^{n-1}(1\hspace*{-.03in}-\hspace*{-.03in}p_i)\right)\log{n\choose 2}\nonumber\\
&\hspace*{-.7in}-n\hspace*{-.02in}\Bigg(\hspace*{-.03in} S(n)\hspace*{-.03in}-\hspace*{-.03in}\frac{3n\hspace*{-.02in}+\hspace*{-.02in}1}{4n}\hspace*{-.02in}+\hspace*{-.02in} n\hspace*{-.03in}\Bigg)p_i(1\hspace*{-.03in}-\hspace*{-.03in}p_i)^{n-1}
\hspace*{-.03in}-\hspace*{-.03in}np_i^{n-1}(1\hspace*{-.03in}-\hspace*{-.03in}p_i)\log(n),
\end{align}
where $S(n)$ is given in Eq.~\eqref{eq_LB_insertion}.
\end{prop}
\begin{IEEEproof}
For the conditional output sequence distribution for a given input sequence, we can write
\begin{align}
    &p(\bfm y|\bfm x(b;n;K))\nonumber\\
    &=\hspace*{-.05in} \left\{\hspace*{-.12in}\begin{array}{cc}
      (1\hspace*{-.02in}-\hspace*{-.02in}p_i)^n &\hspace*{-.05in} \bfm y =\bfm x(b;n;K) \\
      \frac{n_1+1}{4} p_i(1\hspace*{-.02in}-\hspace*{-.02in}p_i)^{n-1} &\hspace*{-.05in} \bfm y=(b;n_1+1,\dotsc  ,n_K) \\
      \frac{n_K+1}{4} p_i(1\hspace*{-.02in}-\hspace*{-.02in}p_i)^{n-1} &\hspace*{-.05in} \bfm y=(b;n_1,\dotsc  ,n_K+1)\\
      \frac{n_k+2}{4} p_i(1\hspace*{-.02in}-\hspace*{-.02in}p_i)^{n-1} &\hspace*{-.05in} \begin{array}{c}\bfm y=(b;n_1,\dotsc  ,n_k+1,\dotsc  ,n_K)\\(1<k<K)\end{array} \\
      \frac{1}{4}p_i(1\hspace*{-.02in}-\hspace*{-.02in}p_i)^{n-1} &\hspace*{-.05in} \bfm y=(b;n_1,\dotsc  ,n'_{k,1},2,n'_{k,2},\dotsc  ,n_K)\\
      \frac{2}{4}p_i(1\hspace*{-.02in}-\hspace*{-.02in}p_i)^{n-1} &\hspace*{-.05in} \bfm y=(b;n_1,\dotsc  ,n''_{k,1},1,n''_{k,2},\dotsc  ,n_K)\\
      \frac{1}{4}p_i(1\hspace*{-.02in}-\hspace*{-.02in}p_i)^{n-1} &\hspace*{-.05in} \bfm y=(\bar{b};1,n_1,\dotsc  ,n_{k},\dotsc  ,n_K)\\
      \frac{1}{4}p_i(1\hspace*{-.02in}-\hspace*{-.02in}p_i)^{n-1} &\hspace*{-.05in} \bfm y=(b;n_1,\dotsc  ,n_{k},\dotsc  ,n_K,1)\\
    \epsilon^i_{y,x}&\hspace*{-.05in} |\bfm y|\geq n+2
    \end{array}\right. \nonumber
\end{align}
where $n'_{k,1}+n'_{k,2}=n_k-1$ ($n'_{k,1}$, $n'_{k,2}\geq0$), $n''_{k,1}+n''_{k,2}=n_k$ ($n''_{k,1}$, $n''_{k,2}\geq1$), and $\epsilon^i_{y,x}$ represents $p(\bfm y|\bfm x(b;n;K))$ for given $\bfm y$ with $|\bfm y|\geq 2$. Furthermore, since there are $n_k$ possibilities for $n'_{k,i}\geq0$ to have ${n'_{k,1}+n'_{k,2}=n_k-1}$, and $n_k-1$ possibilities for $n''_{k,i}\geq1$ to have ${n''_{k,1}+n''_{k,2}=n_k}$, we obtain
\begin{align}
 &H(\bfm Y|\bfm x(b;n;K^x))= -(1-p_i)^n\log(1-p_i)^n\nonumber\\
 &-p_i(1-p_i)^{n-1}\bigg(n\log(p_i(1-p_i)^{n-1})-1.5n-0.5K^x\bigg) \nonumber\\
 &-\hspace*{-.03in}\frac{1}{4}p_i(1\hspace*{-.03in}-\hspace*{-.03in}p_i)^{n-1}\hspace*{-.02in}\Bigg(\hspace*{-.045in}(n^x_1\hspace*{-.03in}+\hspace*{-.03in}1)\log(n^x_1\hspace*{-.03in}+\hspace*{-.03in}1)+(n^x_{\hspace*{-.015in}K^x}\hspace*{-.03in}+\hspace*{-.03in}1)\log(n^x_{\hspace*{-.015in} K^x}\hspace*{-.03in}+\hspace*{-.03in}1)\Bigg.\nonumber\\
 &\Bigg.\hspace*{.95in}+\sum_{k=2}^{K^x-1} (n_k^x+2) \log(n_k^x+2)\Bigg)+H_{\epsilon,i}(\bfm x),\nonumber
\end{align}
where $H_{\epsilon,i}(\bfm x)$ is the term related to the outputs resulting from more than one insertion. Therefore, by considering i.u.d. input sequences, since there are $2{{n-1}\choose {K-1}}$ input sequences of length $n$ with $K$ runs, we have
\begin{align}\label{H(YX)_insertion}
&H(\bfm Y|\bfm X)= -(1-p_i)^n\log(1-p_i)^n+H_{\epsilon,i}(\bfm X)\nonumber\\
&-np_i(1-p_i)^{n-1}\left(\log(p_i(1-p_i)^{n-1})-\frac{7n+1}{4n}+S(n)\right),
\end{align}
where $H_{\epsilon,i}(\bfm X)=\sum_{\bfm x\in \cal{X}}\frac{H_{\epsilon,i}(\bfm x)}{2^n}$ and 
\begin{align}
&S(n)\hspace*{-.03in}=\hspace*{-.03in}\frac{1}{2^{n+2}n}\hspace*{-.03in}\sum_{x,K^x\neq1}\hspace*{-.03in}\bigg[\hspace*{-.03in}(n^x_1\hspace*{-.03in}+\hspace*{-.03in}1)\log(n^x_1\hspace*{-.03in}+\hspace*{-.03in}1)\nonumber\\
&+\hspace*{-.03in}(n^x_{K^x}\hspace*{-.03in}+\hspace*{-.03in}1)\log(n^x_{K^x}\hspace*{-.03in}+\hspace*{-.03in}1)\hspace*{-.03in}+\hspace*{-.03in}\sum_{k=2}^{K^x-1} (n_k^x\hspace*{-.03in}+\hspace*{-.03in}2) \log(n_k^x\hspace*{-.03in}+\hspace*{-.03in}2)\bigg]\hspace*{-.03in}+\hspace*{-.03in}\frac{\log(n)}{2^{n+1}},\nonumber
\end{align}
which can be written as
\begin{align}
&S(n)=\frac{\log(n)}{2^{n+1}}+\frac{1}{2^{n+2}n}\Bigg[\sum_{\bfm x}\sum_{k=1}^{K^x} (n_k^x+2) \log(n_k^x+2)\nonumber\\
&\ +2\sum_{x,K^x\neq1}\left[(n^x_1+1)\log(n^x_1+1)-(n^x_1+2)\log(n^x_1+2)\right]\Bigg]\nonumber\\
&=\frac{1}{4n}\sum_{l=1}^{n-1}2^{-l}\left[(n\hspace*{-.03in}+\hspace*{-.03in}1\hspace*{-.03in}-\hspace*{-.03in}l)(l\hspace*{-.03in}+\hspace*{-.03in}2)\log(l\hspace*{-.03in}+\hspace*{-.03in}2)\hspace*{-.03in}+\hspace*{-.03in}2(l\hspace*{-.03in}+\hspace*{-.03in}1)\log(l\hspace*{-.03in}+\hspace*{-.03in}1)\right]\nonumber\\
&\ +\frac{\log(n)}{2^{n+1}}.
\end{align}
 Here we have used the same approach used in the proof of Proposition~\ref{lemma_Hyx_ub_deletion_substitution}, and considered the fact that there are $2^{n-l}$ sequences of length $n$ with $n_1=l$ or $n_K=l$.

If we assume that all the possible outputs resulting from $k$ insertions ($k\geq 2$) for a given $\bfm x$ are equiprobable, since \begin{equation}
-\sum_{j=1}^J p_j\log(p_j)\leq -\sum_{j=1}^J p_j\log\frac{\sum_{j'=1}^J p_{j'}}{J},
\end{equation}
we can upper bound $H_{\epsilon,i}(\bfm x)$. That is,
\begin{align}
    H_{\epsilon,i}(\bfm x)=&\sum_{k=2}^n \sum_{\bfm y\in {\cal Y}^i(\bfm x+k)} -Q(\bfm y|\bfm x)\log\bigg(Q(\bfm y|\bfm x)\bigg)\nonumber\\
    \leq& \sum_{k=2}^n\hspace*{-.01in}-\hspace*{-.01in}\epsilon_k\hspace*{-.01in}\log\bigg(\hspace*{-.01in}\frac{\epsilon_k}{|{\cal Y}^i(\bfm x+k)|}\hspace*{-.01in}\bigg)\hspace*{-.01in} \leq\hspace*{-.01in} \sum_{k=2}^n \hspace*{-.01in}-\hspace*{-.01in}\epsilon_k\hspace*{-.01in}\log\bigg(\hspace*{-.01in}\frac{\epsilon_k}{2^{n+k}}\hspace*{-.01in}\bigg),\nonumber
   \end{align}
where $\epsilon_k=\sum_{\bfm y\in (\bfm x,k)} Q(\bfm y|\bfm x)={n\choose k}p_i^k(1-p_i)^{n-k}$ is the probability of $k$ insertions in transmission of $n$ bits, and the last inequality follows since $|{\cal{Y}}^i(\bfm x+k)|\leq2^{n+k}$, where $|{\cal{Y}}^i(\bfm x+k)|$ denotes the number of output sequences resulting from $k$ insertions into a given input sequence $\bfm x$. After some algebra, we arrive at
\begin{align}\label{h_epsilonX_insertion}
    H_{\epsilon,i}(\bfm X) \leq & \
    n(1+p_i)+n H_b(p_i)-n(1-p_i)^n\nonumber\\
    &-(n+1)np_i(1-p_i)^{n-1}\hspace*{-.02in}+ \hspace*{-.02in}(1-p_i)^n\log(1-p_i)^n\nonumber\\
    &+np_i(1-p_i)^{n-1}\log \left({p_i(1-p_i)^{n-1}}\right)\nonumber\\
    &-np_i^{n-1}(1-p_i)\log(n)-\bigg(1-p_i^n-(1-p_i)^n\nonumber\\
    &-np_i(1-p_i)^{n-1}-np_i^{n-1}(1-p_i)\bigg)\log{n\choose 2}.\nonumber
\end{align}
Finally, by substituting the above upper bound into Eq.~\eqref{H(YX)_insertion}, the upper bound~\eqref{eq_H(Y|X)_ins} is obtained.
\end{IEEEproof} 
\vspace*{.1in}

\textit{\textbf{Proof of Lemma~\ref{lemma_C_insertion}}}:
By substituting the exact value of the output entropy (Eq.~\eqref{eq_HY_insertion}) and the upper bound on the conditional output entropy (Eq.~\eqref{eq_H(Y|X)_ins}) of the random insertion channel with i.u.d. input sequences into Eq.~\eqref{I}, a lower bound on the achievable information rate is obtained, hence the lemma is proved.\hfill$\Box$

\section{Numerical Examples}\label{numerical_ex}
We now present several examples of the lower bounds on the insertion/deletion channel capacity for different values of $n$ and compare them with the existing ones in the literature.
\subsection{Deletion-Substitution Channel}

In Table~\ref{tb_del_sub}, we compare the lower bound~\eqref{eq_lb_del_sub} for $n=100$ and $n=1000$ with the one in~\cite{gallager}. We observe that the new bound improves the result of~\cite{gallager} for the entire range of $p_d$ and $p_e$, and also as expected, by increasing $n$ from $100$ to $1000$, a tighter lower bound for all values of $p_d$ and $p_e$ is obtained.
\begin{table*}[ht]
\caption{Lower bounds on the capacity of the deletion-substitution channel (In the left hand side table ``1-lower bound'' is reported).}
\begin{minipage}[b]{0.6\linewidth}\centering
\scalebox{.95}{\begin{tabular}{|c|c||c|c|c|c|}
\hline
        $p_d$ &$p_e$ & $1-$LB~\eqref{LB_gallager_del_sub} & $1-$LB~\eqref{eq_lb_del_sub}  &$1-$LB~\eqref{eq_lb_del_sub}  \\
        &&&$n=1000$&$n=100$\\
        \hline						
        \hline
        $10^{-5}$ &$10^{-5}$&$3.6104\times10^{-4}$&	$\bf{3.5817\times10^{-4}}$&	$3.5834\times10^{-4}$ \\
        \hline
        $10^{-5}$ &$10^{-4}$& $1.6535\times10^{-3}$&	$\bf{1.6506\times10^{-3}}$&	 $1.6508\times10^{-3}$\\
        \hline
        $10^{-5}$ &$10^{-3}$& $1.15881\times10^{-2}$&	$\bf{ 1.15853\times10^{-2}}$&$1.15854\times10^{-2}$\\
        \hline
        $10^{-4}$ &$10^{-5}$& $1.6535\times10^{-3}$	&$\bf{1.6248\times10^{-3}}$	&$1.6264\times10^{-3}$ \\
        \hline
        $10^{-4}$ &$10^{-4}$& $2.9459\times10^{-3}$	&$\bf{2.9172\times10^{-3}}$&	 $2.9188\times10^{-3}$\\
        \hline
        $10^{-4}$ &$10^{-3}$& $1.2879\times10^{-2}$&$\bf{1.2850\times10^{-2}}$&	 $1.2852\times10^{-2}$\\
        \hline
        $10^{-3}$ &$10^{-5}$&$1.1588\times10^{-2}$&	$\bf{1.1302\times10^{-2}}$&	$1.1319\times10^{-2}$ \\
        \hline
        $10^{-3}$ &$10^{-4}$&$1.2879\times10^{-2}$&	$\bf{1.2593\times10^{-2}}$&	$1.261\times10^{-2}$ \\
        \hline
        $10^{-3}$ &$10^{-3}$&$2.2804\times10^{-2}$	&$\bf{2.2518\times10^{-2}}$&	$2.2535\times10^{-2}$ \\
        \hline
        \end{tabular}}
\end{minipage}
\begin{minipage}[b]{0.4\linewidth}
\centering
\scalebox{.95}{\begin{tabular}{|c|c||c|c|c|c|}
\hline
$p_d$ &$p_e$ & LB~\eqref{LB_gallager_del_sub} & LB~\eqref{eq_lb_del_sub}  &LB~\eqref{eq_lb_del_sub} \\
        &&&$n=1000$&$n=100$\\
        \hline						
        \hline
         0.01&0.01&0.8392 & \bf{0.8419}& 0.8418\\
        \hline
        0.01&0.03&0.7268 &\bf{0.7373} &0.7293\\
        \hline
        0.01&0.10&0.4549 &  \bf{0.4576}& 0.4575 \\
        \hline
        0.05&0.01&0.6368 & \bf{0.6476}& 0.6469\\
        \hline
        0.05&0.03& 0.5289 & \bf{0.5397}& 0.5390\\
        \hline
        0.05&0.10&0.2681 &\bf{0.2789} & 0.2781\\
        \hline
        0.10&0.01& 0.4583&\bf{0.4729} & 0.4716\\
        \hline
        0.10&0.03&0.3561 & \bf{0.3707}& 0.3693\\
        \hline
        0.10&0.10&0.1089 & \bf{0.1236}& 0.1222\\
        \hline
\end{tabular}}
\end{minipage}
\label{tb_del_sub}
\end{table*}

\subsection{Deletion-AWGN Channel}
 We now compare the derived analytical lower bound on the capacity of the deletion-AWGN channel with the simulation based bound of~\cite{junhu} which is the achievable information rate of the deletion-AWGN channel for i.u.d. input sequences obtained by Monte-Carlo simulations. As we observe in Fig.~\ref{fig:del_AWGN2}, the lower bound~\eqref{LB_del_AWGN} is very close to the simulation results of~\cite{junhu} for small values of deletion probability but it does not improve them. This is not unexpected, because we further lower bounded the achievable information rate for i.u.d. input sequences while in~\cite{junhu}, the achievable information rate for i.u.d. input sequences is obtained by Monte-Carlo simulations without any further lower bounding. On the other hand, new bound is provable, analytical and very easy to compute while the result in~\cite{junhu} requires lengthly simulations. Furthermore, the procedure employed in~\cite{junhu} is only useful for deriving lower bounds for small values of deletion probability, e.g., $p_d\leq0.1$, while the lower bound~\eqref{LB_del_AWGN} holds for a much wider range.

\begin{figure}
    \centering
    \includegraphics[width=.48\textwidth]{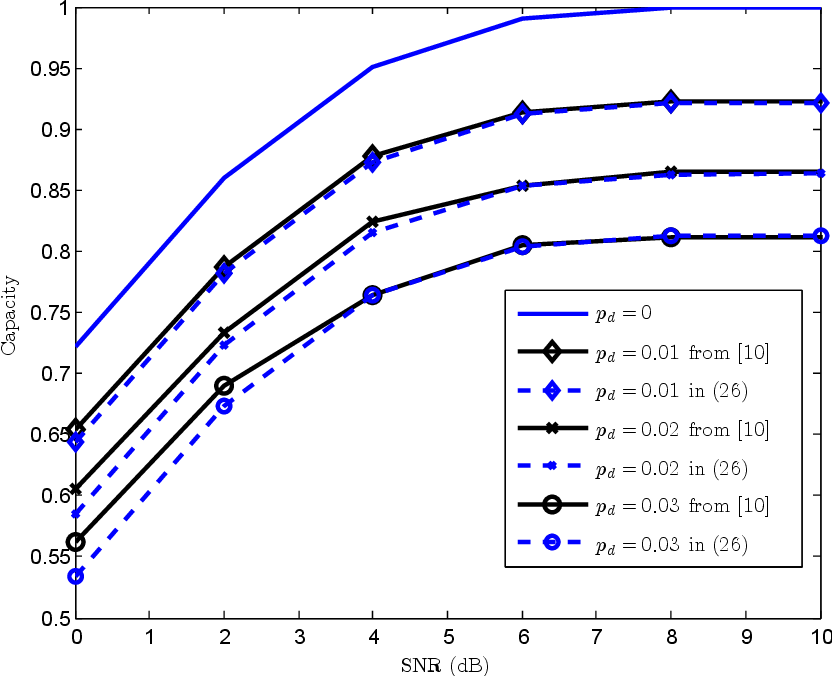}
    \caption{Comparison between the lower bound~\eqref{LB_del_AWGN} for $n=1000$ with the lower bound in~\cite{junhu} versus SNR for different deletion probabilities.}
    \label{fig:del_AWGN2}
\end{figure} 
\subsection{Random Insertion Channel}
We now numerically evaluate the lower bounds derived on the capacity of the random insertion channel. Similar to the previous cases, different values of $n$ result in different lower bounds. In Table~\ref{tb_insertion_new} and Fig.~\ref{fig_ins_rnd_new}, we compare the lower bound in Eq.~\eqref{eq_LB_insertion} with  the lower bound due to Gallager~\cite{gallager} $C_i\geq 1-H_b(p_i)$, where the reported values are obtained for the optimal value of $n$. 
\begin{table*}[ht]
\caption{Lower bounds on the capacity of the random insertion channel (In the left hand side table ``1-lower bound'' is reported).}		
    \label{tb_insertion_new}
\begin{minipage}[b]{0.55\linewidth}\centering
    \centering 			
    \scalebox{.95}{\begin{tabular}{|c||c|c|c|}    
        \hline
        $p_i$ & $1-$LB from~\cite{gallager} & $1-$LB~\eqref{eq_LB_insertion}& optimal \\
              &								&								& value of $n$\\
        \hline        
        \hline
        $10^{-6}$ & $2.14\times10^{-5}$&$	\bf{2.007\times10^{-5}}$& 121	 \\
       \hline
        $10^{-5}$ & $1.81\times10^{-4}$&$	\bf{1.68\times10^{-4}}$ & 57\\
        \hline
        $10^{-4}$ & $1.47\times10^{-3}$& $	\bf{1.35\times10^{-3}}$& 27\\	 
        \hline
        $10^{-3}$ &$1.14\times10^{-2}$& $\bf{1.02\times10^{-2}}$ & 13 \\
        \hline
		$10^{-2}$ & $8.07\times10^{-1}$&$\bf{7.14\times10^{-2}}$&7\\
		\hline        
    \end{tabular}}
\end{minipage}
\begin{minipage}[b]{0.4\linewidth}\centering

    \scalebox{.95}{\begin{tabular}{|c||c|c|c|}
        \hline
        $p_i$ & LB from~\cite{gallager} & LB~\eqref{eq_LB_insertion}  & optimal\\
        & & &  value of $n$ \\
       \hline
        \hline
        0.03 & 0.8056 & \bf{0.8276} & 5 \\
        \hline
        0.05 & 0.7136 & \bf{0.7442} & 5 \\
        \hline
        0.10 & 0.5310 & \bf{0.5702} & 4 \\
        \hline
        0.15 & 0.3901 & \bf{0.4230} & 4 \\
        \hline
        0.20 & 0.2781 & \bf{0.2962} & 3 \\
        \hline
        0.23 & 0.2220 & \bf{0.2283} & 3 \\
        \hline
        0.25 & \bf{0.1887} & 0.1853 & 3 \\
                \hline
    \end{tabular}}

\end{minipage}
\end{table*}
\begin{figure}
    \includegraphics[width=.49\textwidth]{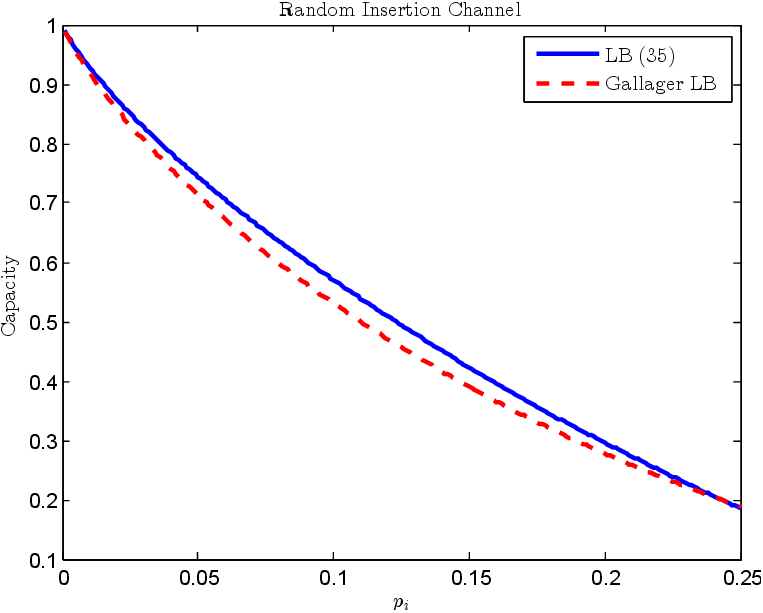}
    \caption{Comparison of the lower bound~\eqref{eq_LB_insertion} with lower bound presented in~\cite{gallager}.}
    \label{fig_ins_rnd_new}
\end{figure}
We observe that for larger $p_i$, smaller values of $n$ give the tightest lower bounds. This is not unexpected since in upper bounding $H(\bfm Y|\bfm X)$, we computed the exact value of $p(\bfm y|\bfm x)$ for at most one insertion, i.e., $|\bfm y|=|\bfm x|$ or $|\bfm y|=|\bfm x|+1$, and upper bounded the part of the conditional entropy resulting form  more than one insertion. Therefore, for a fixed $p_i$ by increasing $n$, the probability of having more than one insertion increases and as a result the upper bound becomes loose. We also observe that the lower bound~\eqref{eq_LB_insertion} improves upon the lower bound in~\cite{gallager} for $p_i<0.25$, e.g., for $p_i=0.1$, we achieve an improvement of $0.0392$ bits/channel use.

\section{Acknowledgments}

We would like to thank the editor and the reviewers for detailed comments on the manuscript. In particular, we would like to acknowledge that the simpler proof of Lemma 1 given in Appendix A is due to one of the reviewers. 

\section{Conclusions}\label{conclusion}
We have presented several analytical lower bounds on the capacity of the insertion/deletion channels by lower bounding the mutual information rate for i.u.d. input sequences. We have derived the first analytical lower bound on the capacity of the deletion-AWGN channel which for small values of deletion probability is very close to the existing simulation based lower bounds. The lower bound presented on the capacity of the deletion-substitution channel improves the existing analytical lower bound for all values of deletion and substitution probabilities. For random insertion channel, the presented lower bound improve the existing ones for $p_i<0.25$. For $p_e=0$, the presented lower bound on the capacity of the deletion-substitution channel results into a lower bound on the capacity of the deletion-only channel which for small values of deletion probability, is very close to the tightest presented lower bounds, and is in agreement with the first order expansion of the channel capacity for $p_d\rightarrow0$, while our result is a strict lower bound for the entire range of $p_d$.

\appendices

\section{Deletion-substitution channel capacity in terms of the deletion channel capacity}\label{app-reviewer}
In this appendix, we relate the deletion-substitution and deletion-only channel capacities through an inequality (as pointed to us by one of the reviewers) which is a special case of a result obtained by the authors in~\cite{ISIT-noisy}. This inequality can provide a tool to provide simpler proof for Lemma~\ref{lemma_C_del_sub}.

\begin{claim}
For any possible input distribution $P(\bfm X)$, we have
\begin{equation}\label{reviewer-claim}
I(\bfm X;\bfm Y')\geq I(\bfm X;\bfm Y)-n(1-p_d)H_b(p_e).
\end{equation}
\end{claim}

\begin{IEEEproof}
In Fig.~\ref{fig:del_sub}, $\bfm X\to\bfm Y\to\bfm Y'$ form a Markov chain. Let $\bfm F$ be the `flipping' process of the BSC channel, consisting of $(1-p_d +\delta)n$ bits drawn from i.i.d. Bernoulli($p_e$), where a 1 represents a flip, and 0 represents a location that is unaffected, and $\delta > 0$ is some constant we can choose later. Clearly, $\bfm Y' = f(\bfm Y ;\bfm F)$ with high probability for the obvious function $f(.)$ which does $Y'_i = Y_i\oplus F_i$ for all bits in $\bfm Y$. (There is a problematic event corresponding to more than $(1-p_d + \delta)n$ bits passing through the deletion channel, but the probability of this event goes to 0 as $n\to \infty$. This event can be dealt with and we ignore it below, simply assuming $\bfm Y'= f(\bfm Y ;\bfm F)$. Note that we also have $\bfm Y = f(\bfm Y';\bfm F)$ at the same time).

Hence, for the mutual information $I(\bfm X;\bfm Y')$, we have
\begin{eqnarray}
I(\bfm X;\bfm Y')&=&H(\bfm X)-H(\bfm X|\bfm Y')\nonumber\\
&=&H(\bfm X)-H(\bfm X|\bfm Y',\bfm F)-I(\bfm F;\bfm X|\bfm Y').\nonumber
\end{eqnarray}
Now, $H(\bfm X|\bfm Y',\bfm F)=H(\bfm X|\bfm Y,\bfm Y',\bfm F)=H(\bfm X|\bfm Y)$ since $\bfm Y=f(\bfm Y',\bfm F)$ and $\bfm X\to\bfm Y\to\bfm (\bfm F,\bfm Y')$ form a Markov chain. Further, $I(\bfm F;\bfm X|\bfm Y')\leq H(\bfm F|\bfm Y')=H(\bfm F)=n{(1-p_d+\delta)}H_b(p_e)$. It follows that
\begin{eqnarray}
I(\bfm X;\bfm Y')&\geq& H(\bfm X)-H(\bfm X|\bfm Y)-n(1-p_d+\delta)H_b(p_e)\nonumber\\
&=&I(\bfm X;\bfm Y)-n(1-p_d+\delta)H_b(p_e).\nonumber
\end{eqnarray}
Since $\delta>0$ is arbitrary, the result follows.
\end{IEEEproof}
\begin{cor}
Let $C_d$ and $C_{ds}$ denote the deletion-only and deletion-substitution channel capacities, respectively, then
\begin{equation}\label{reviewer-cor}
C_{ds}\geq C_d -(1-p_d)H_b(p_e).
\end{equation}
\end{cor}
\begin{IEEEproof}
Since Eq.~\eqref{reviewer-claim} holds for any possible input distribution, it holds for capacity achieving input distribution for the deletion-only channel as well. Therefore, by dividing both sides by $n$ and letting $n$ go to infinity the proof follows.
\end{IEEEproof}

\section{Part of Proof of Proposition~\ref{lemma_Hyx_ub_deletion_substitution}}\label{app_hyx_del_sub}

\small\begin{align}
&H\bigg(\bfm Y'\bigg|\bfm x(b;n;K^x)\bigg)\nonumber\\
&=-\sum_{j=0}^n\sum_{\bfm y'\in {{\cal{Y}}{_{-j}^d}}}P\left(\bfm y'(n-j)|\bfm x\right)\log \left(P\left(\bfm y'(n-j)|\bfm x\right)\right)\nonumber\\
&\leq\hspace*{-.02in} -\hspace*{-.03in}\sum_{j=0}^{n}\hspace*{-.015in}\sum_{\bfm y'\in {{\cal{Y}}{_{\hspace*{-.03in}-j}^d}}}\hspace*{-.03in}\sum_{D\in {\cal{D}}{_K^n}(j)}\hspace*{-.13in} P(\bfm y'|D\hspace*{-.03in}*\hspace*{-.03in}\bfm x)P(D|\bfm x)\hspace*{-.015in}\log\hspace*{-.015in}\left(P(\bfm y'|D\hspace*{-.03in}*\hspace*{-.03in}\bfm x)P(D|\bfm x)\right)\hspace*{-.03in},\nonumber
\end{align}\normalsize
where the inequality is obtained from the expression in~\eqref{eq_split}. Furthermore, by employing the results from Eqs.~\eqref{Q(y|x(n,K))_del_sub} and~\eqref{eq_bsc} and using the fact that there are ${{n-j}\choose s}$, distinct output sequences of length $n-j$ resulting from $s$ substitution errors into a given input $\bfm x$, i.e., $s=d_H\left(\bfm y'(n-j);D(n;K;j)*\bfm x(n;K)\right)$, we arrive at
\begin{align}\label{eq_H(Y'|x)}
&H\bigg(\bfm Y'\bigg|\bfm x(b;n;K^x)\bigg)\nonumber\\
&\leq -\sum_{j=0}^n\sum_{s=0}^{n-j}{{n-j}\choose s}\sum_{j_1+\cdots +j_K=j}p_e^s(1-p_e)^{n-j-s}\times\nonumber\\
&\quad\quad\times{{n^x_1}\choose j_1}\cdots {{n^x_K}\choose j_K} p_d^j(1-p_d)^{n-j}\times\nonumber\\
&\quad\quad \times\log\bigg({{n^x_1}\choose j_1}\cdots{{n^x_K}\choose j_K} p_d^j(1-p_d)^{n-j}p_e^s(1-p_e)^{n-j-s}\bigg)\nonumber\\
&=-\sum_{j=0}^n \sum_{j_1+\cdots +j_K=j}{{n^x_1}\choose j_1}\cdots{{n^x_K}\choose j_K}
p_d^j(1-p_d)^{n-j}\times\nonumber\\
&\times\hspace*{-.02in}\bigg[\hspace*{-.02in}\hspace*{-.035in}-(n\hspace*{-.02in}-\hspace*{-.02in} j) H_b(p_e)\hspace*{-.02in}+\hspace*{-.02in}\log\bigg({{n^x_1}\choose j_1}\cdots {{n^x_K}\choose j_K} p_d^j(1-p_d)^{n-j}\bigg)\bigg]\nonumber\\
&=n H_b(p_d)\hspace*{-.02in} -\hspace*{-.02in}\sum_{j=0}^n\sum_{j_1+\cdots +j_K=j}\hspace*{-.06in}{{n^x_1}\choose j_1}\cdots {{n^x_K}\choose j_K} p_d^j(1\hspace*{-.02in}-\hspace*{-.02in}p_d)^{n-j}\times\nonumber\\
&\hspace*{.45in} \times \bigg[-n(1-p_d) H_b(p_e)+\log\bigg({{n^x_1}\choose j_1}\cdots {{n^x_K}\choose j_K}\bigg)\bigg].\nonumber
\end{align}
Using the generalized Vandermonde's identity, that is,
$$\sum_{j_1+\dotsc  +j_{K^x}=j}{{n_1^x}\choose j_1}\dotsc  {{n_{K^x}^x}\choose
j_{K^x}}={n\choose j},$$ and the result
\begin{align}
\sum_{j_1+\dotsc  +j_{K^x}=j}&{{n_1^x}\choose j_1}\dotsc  {{n_{K^x}^x}\choose
j_{K^x}} \log\left({{n_1^x}\choose j_1}\dotsc  {{n_{K^x}^x}\choose
j_{K^x}}\right)\nonumber\\
&=\sum_{j_1+\dotsc  +j_{K^x}=j}{{n_1^x}\choose
j_1}\dotsc  {{n_{K^x}^x}\choose j_{K^x}}\sum_{k=1}^{K^x}
\log{{n_k^x}\choose j_k}\nonumber\\
&=\sum_{k=1}^{K^x}\sum_{j_k=0}^j {{n_k^x}\choose
j_k}{{n-n_k^x}\choose {j-j_k}} \log{{n_k^x}\choose j_k}\nonumber,
\end{align}
we obtain
\begin{align}
  &H\bigg(\bfm Y'\bigg|\bfm x(b;n;K^x)\bigg) \leq n H_b(p_d)+n(1-p_d)H_b(p_e)\nonumber\\
  &-\sum_{j=0}^n p_d^j(1-p_d)^{n-j}\sum_{k=1}^{K^x}\sum_{j_k=0}^j {{n_k^x}\choose j_k}{{n-n_k^x}\choose {j-j_k}} \log{{n_k^x}\choose j_k}.\nonumber
\end{align}

\section{Proof of Proposition~\ref{lemma_Hyx_UB_del_AWGN}}\label{app_H(YX)_ub_awgn}
For an i.i.d. deletion-AWGN channel, for a given $\bfm x(b;n;K)$ and a fixed $j$, defining $\alpha(D,\bfm x)=1-2(D*\bfm x)$, i.e., $\alpha_i(D,\bfm x)\in \{1,-1\}$, yields
\begin{align}
&f_{\widetilde{\bfm y}}(\eta|\bfm x(b;n;K),j)\nonumber\\
&=\hspace*{-.05in}\sum_{D \in {\cal{D}}{_K^n}(j)}\hspace*{-.045in} f_{\widetilde{\bfm y}}(\eta|\bfm x(b;n;K),D)P(D|\bfm x(b;n;K))\nonumber\\
&=\hspace*{-.05in}\sum_{D\in {\cal{D}}{_K^n}(j)}\hspace*{-.045in}f_{\widetilde{\bfm y}}(\eta|\alpha(D,\bfm x))P(D|\bfm x(b;n;K))\nonumber\\
&=\hspace*{-.05in}\sum_{D\in{\cal{D}}{_K^n}(j)}\hspace*{-.045in}f_{\widetilde{y}_1\dotsc  \widetilde{y}_{n-j}}(\eta_1\dotsc  \eta_{n-j}|\alpha_1\dotsc  \alpha_{n-j})P(D|\bfm x(b;n;K))\nonumber\\
&=\hspace*{-.05in}\sum_{D \in {\cal{D}}{_K^n}(j)}\hspace*{-.045in}f_{\widetilde{y}_1}(\eta_1|\alpha_1)\dotsc  f_{\widetilde{y}_{n-j}}(\eta_{n-j}|\alpha_{n-j})P(D|\bfm x(b;n;K)),\nonumber
\end{align}
where the last equality follows the fact that the noise samples $\bfm z_i$ are independent and $\alpha_i(D,\bfm x)$ are also independent. By employing
\begin{equation}
f_{\widetilde{y}_i}(\eta_i|\alpha_i(D,\bfm x))=\frac{1}{\sqrt{2\pi}\sigma}\exp\left(\frac{-(\eta_i-\alpha_i(D,\bfm x))^2}{2\sigma^2}\right)\nonumber,
\end{equation}
and $\displaystyle P\left(D(n;K;j)\bigg|\bfm x(b;n;K),j\right)=\frac{{n_1\choose j_1}\dotsc  {n_K\choose j_K}}{{n\choose j}}$, we can write
\begin{align}
&f_{\widetilde{\bfm y}}(\eta|\bfm x(b;n;K),j)\nonumber\\
&=\frac{1}{(\sqrt{2\pi}\sigma)^{n-j}}\hspace*{-.07in}\sum_{D\in {\cal{D}}{_K^n}(j)}\hspace*{-.03in}\prod_{i=1}^{n-j}e^{\frac{-\left(\eta_i-\alpha_i(D,\bfm x)\right)^2}{2\sigma^2}} P\left(D|\bfm x(b;n;K),j\right)\nonumber\\
&=\frac{1}{(\sqrt{2\pi}\sigma)^{n-j}}\hspace*{-.07in}\sum_{j_1+\dotsc  +j_K=j}\frac{{n_1\choose j_1}\dotsc  {n_K\choose j_K}}{{n\choose j}}\prod_{i=1}^{n-j}e^{\frac{-\left(\eta_i-\alpha_i(D,\bfm x)\right)^2}{2\sigma^2}},\nonumber
\end{align}
Therefore, by defining $$A(j_1,\dotsc  ,j_K)=\frac{{n_1\choose j_1}\dotsc  {n_K\choose j_K}}{{n\choose j}}\prod_{i=1}^{n-j}e^{\frac{-\left(\eta_i-\alpha_i(D,\bfm x)\right)^2}{2\sigma^2}},$$ we obtain
\begin{align}
&h(\widetilde{\bfm Y}|\bfm x,j)\nonumber\\
&=-\int_{-\infty}^{\infty}\dotsc  \int_{-\infty}^{\infty} \frac{1}{(\sqrt{2\pi}\sigma)^{n-j}}\sum_{j_1+\dotsc  +j_K=j}A(j_1,\dotsc  ,j_K)\times\nonumber\\
&\times\hspace*{-.02in}\log\left(\hspace*{-.02in}\frac{1}{(\sqrt{2\pi}\sigma)^{n-j}}\sum_{j'_1+\dotsc  +j'_K=j}A(j'_1,\dotsc  ,j'_K)\hspace*{-.02in}\right)\hspace*{-.02in}d\eta_1\dotsc  d\eta_{n-j}\nonumber\\
&=-\int_{-\infty}^{\infty}\dotsc  \int_{-\infty}^{\infty} \frac{1}{(\sqrt{2\pi}\sigma)^{n-j}}\sum_{j_1+\dotsc  +j_K=j}A(j_1,\dotsc  ,j_K)\times\nonumber\\
&\times\left[\log\left(\sum_{j'_1+\dotsc  +j'_K=j}A(j'_1,\dotsc  ,j'_K)\right)\right]d\eta_1\dotsc  d\eta_{n-j}\nonumber\\
&+(n-j)\log(\sqrt{2\pi}\sigma),\nonumber
\end{align}
where we used the result of the generalized Vandermonde's identity and also the fact that $\int_{-\infty}^\infty f_{\widetilde{y}_i}(\eta_i|\bar{y}_i)d\eta_i=1$. By using the inequality
\begin{align} 
\sum_{j'_1+\dotsc  +j'_K=j}A(j'_1,\dotsc  ,j'_K)\geq A(j_1,\dotsc  ,j_K)\nonumber,
\end{align}
which holds for every $j_1+\dotsc  +j_K=j$, we can write
\begin{align}
&h(\widetilde{\bfm Y}|\bfm x,j)\leq(n-j)\log(\sqrt{2\pi}\sigma)\nonumber\\
&-\int_{-\infty}^{\infty}\dotsc  \int_{-\infty}^{\infty} \frac{1}{(\sqrt{2\pi}\sigma)^{n-j}}\sum_{j_1+\dotsc  +j_K=j}A(j_1,\dotsc  ,j_K)\times\nonumber\\
&\hspace*{1.1in}\times\log\left(A(j_1,\dotsc,j_K)\right)d\eta_1\dotsc  d\eta_{n-j}\nonumber\\
&=(n-j)\log(\sqrt{2\pi e}\sigma)+\log{n\choose j}\nonumber\\
&\ -\sum_{j_1+\dotsc  +j_K=j}\frac{{n_1\choose j_1}\dotsc  {n_K\choose j_K}}{{n\choose j}}\log\left({n_1\choose j_1}\dotsc  {n_K\choose j_K}\right).\nonumber
\end{align}
By considering i.u.d. input sequences, we have
\begin{align}\label{eq_H(Y|XT)}
h(\widetilde{\bfm Y}|\bfm X,T)=&\ \sum_{j=0}^n {n\choose j}p_d^j(1-p_d)^{n-j}\sum_{\bfm x\in \cal X}\frac{1}{2^n}h(\widetilde{\bfm Y}|\bfm x,j)\nonumber\\
\leq &\ n(1-p_d)\log(\sqrt{2\pi e}\sigma)\nonumber\\
&+\hspace*{-.02in}\sum_{j=0}^n {n\choose j}p_d^j(1-p_d)^{n-j}\left[\log{n\choose j}\hspace*{-.02in}-\hspace*{-.02in}W_j(n)\right]\hspace*{-.02in},
\end{align}
where $W_j(n)$ is given in Eq.~\eqref{eq_LB_deletion}, and the result is obtained by following the same steps as in the computation leading to~\eqref{Hhs}. Therefore, by substituting Eq.~\eqref{eq_H(Y|XT)} into Eq.~\eqref{eq_H(Yj|Xj)}, Eq.~\eqref{Hyx_UB_del_AWGN} is obtained which concludes the proof.


\begin{thebibliography}{10}
\providecommand{\url}[1]{#1}
\csname url@samestyle\endcsname
\providecommand{\newblock}{\relax}
\providecommand{\bibinfo}[2]{#2}
\providecommand{\BIBentrySTDinterwordspacing}{\spaceskip=0pt\relax}
\providecommand{\BIBentryALTinterwordstretchfactor}{4}
\providecommand{\BIBentryALTinterwordspacing}{\spaceskip=\fontdimen2\font plus
\BIBentryALTinterwordstretchfactor\fontdimen3\font minus
  \fontdimen4\font\relax}
\providecommand{\BIBforeignlanguage}[2]{{%
\expandafter\ifx\csname l@#1\endcsname\relax
\typeout{** WARNING: IEEEtran.bst: No hyphenation pattern has been}%
\typeout{** loaded for the language `#1'. Using the pattern for}%
\typeout{** the default language instead.}%
\else
\language=\csname l@#1\endcsname
\fi
#2}}
\providecommand{\BIBdecl}{\relax}
\BIBdecl

\bibitem{dobrushin}
R.~L. Dobrushin, ``{Shannon's theorems for channels with synchronization
  errors},'' \emph{Problems of Information Transmission}, vol.~3, no.~4, pp.
  11--26, 1967.

\bibitem{dobrushin_general}
------, ``{General formulation of Shannon's main theorem on information
  theory},'' \emph{American Math. Soc. Trans.}, vol.~33, pp. 323--438, 1963.

\bibitem{gallager}
R.~Gallager, ``{Sequential decoding for binary channels with noise and
  synchronization errors},'' \emph{Tech. Rep., MIT Lincoln Lab. Group Report},
  1961.

\bibitem{diggavi2001transmission}
S.~Diggavi and M.~Grossglauser, ``{On transmission over deletion channels},''
  in \emph{Proceedings of the Annual Allerton Conference on Communication
  Control and Computing}, vol.~39, no.~1, 2001, pp. 573--582.

\bibitem{diggavi2006information}
------, ``{On information transmission over a finite buffer channel},''
  \emph{IEEE Transactions on Information Theory}, vol.~52, no.~3, pp.
  1226--1237, 2006.

\bibitem{drinea2006lower}
E.~Drinea and M.~Mitzenmacher, ``{On lower bounds for the capacity of deletion
  channels},'' \emph{IEEE Transactions on Information Theory}, vol.~52, no.~10,
  pp. 4648--4657, 2007.

\bibitem{drinea2007improved}
------, ``Improved lower bounds for i.i.d. deletion and insertion channels,''
  \emph{IEEE Transactions on Information Theory}, vol.~53, no.~8, pp.
  2693--2714, 2007.

\bibitem{drinea2007}
A.~Kirsch and E.~Drinea, ``Directly lower bounding the information capacity for
  channels with i.i.d. deletions and duplications,'' \emph{IEEE Transactions on
  Information Theory}, vol.~56, no.~1, pp. 86 --102, 2010.

\bibitem{kavcic2004insertion}
A.~Kavcic and R.~H. Motwani, ``{Insertion/deletion channels: Reduced-state
  lower bounds on channel capacities},'' in \emph{Proceedings of IEEE
  International Symposium on Information Theory (ISIT)}, 2004, p. 229.

\bibitem{junhu}
J.~Hu, T.~M. Duman, M.~F. Erden, and A.~Kavcic, ``{Achievable information rates
  for channels with insertions, deletions and intersymbol interference with
  i.i.d. inputs},'' \emph{IEEE Transactions on Communications}, vol.~58, no.~4,
  pp. 1102--1111, 2010.

\bibitem{dario}
D.~Fertonani and T.~M. Duman, ``Novel bounds on the capacity of the binary
  deletion channel,'' \emph{IEEE Transactions on Information Theory}, vol.~56,
  no.~6, pp. 2753--2765, 2010.

\bibitem{diggavi-capacity}
S.~Diggavi, M.~Mitzenmacher, and H.~Pfister, ``{Capacity upper bounds for
  deletion channels},'' in \emph{Proceedings of the International Symposium on
  Information Theory (ISIT)}, 2007, pp. 1716--1720.

\bibitem{dario2}
D.~Fertonani, T.~M. Duman, and M.~F. Erden, ``{Bounds on the capacity of
  channels with insertions, deletions and substitutions},'' \emph{IEEE
  Transactions on Communications}, vol.~59, no.~1, pp. 2--6, 2011.

\bibitem{kanoria}
Y.~Kanoria and A.~Montanari, ``{On the deletion channel with small deletion
  probability},'' in \emph{Proceedings of the International Symposium on
  Information Theory (ISIT)}, June 2010, pp. 1002--1006.

\bibitem{asymptotic}
A.~Kalai, M.~Mitzenmacher, and M.~Sudan, ``Tight asymptotic bounds for the
  deletion channel with small deletion probabilities,'' in \emph{Proceedings of
  the International Symposium on Information Theory (ISIT)}, June 2010, pp. 997
  --1001.

\bibitem{ISIT-noisy}
M.~Rahmati and T.~M. Duman, ``On the capacity of binary input symmetric q-ary
  output channels with synchronization errors,'' in \emph{Proceedings of the
  International Symposium on Information Theory (ISIT)}, July 2012, pp.
  691--695.

\bibitem{proakis}
J.~G. Proakis, \emph{{ Digital Communications}}.\hskip 1em plus 0.5em minus
  0.4em\relax 5th ed, New York: McGraw-Hill, 2007.

\end{thebibliography}

\begin{IEEEbiographynophoto}{Mojtaba Rahmati}
(S'11) received his B.~S. degree in electrical engineering in 2007 from University of Tehran, Iran, his M.~S. degree in telecommunication systems in 2009 from Sharif University of Technology, Tehran, Iran and his PhD in electrical engineering in 2013 from Arizona State University, Tempe, AZ. He is currently a research assistant at Arizona State University. His research interests include information theory, digital communications and digital signal processing.
\end{IEEEbiographynophoto}
\begin{IEEEbiographynophoto}{Tolga M.~Duman}
(S'95--M'98--SM'03--F'11) is a Professor of Electrical and Electronics Engineering Department of Bilkent University in Turkey, and is on leave from the School of ECEE at Arizona State University. He received the B.S. degree from Bilkent University in Turkey in 1993, M.S. and Ph.D. degrees from Northeastern University, Boston, in 1995 and 1998, respectively, all in electrical engineering. Prior to joining Bilkent University in September 2012, he has been with  the Electrical Engineering Department of Arizona State University first as an Assistant Professor (1998-2004), then as an Associate Professor (2004-2008), and starting August 2008 as a Professor. Dr. Duman's current research interests are in systems, with particular focus on communication and signal processing, including wireless and mobile communications, coding/modulation, coding forwireless communications, data storage systems and underwater acoustic communications.

Dr. Duman is a Fellow of IEEE, a recipient of the National Science Foundation CAREER Award and IEEE Third Millennium medal. His publications
include a book on MIMO Communications (by Wiley in 2007), overfifty journal papers and over one hundred conference papers. He served as an
editor for IEEE Trans. on Wireless Communications (2003-08), IEEE Trans. on Communications (2007-2012) and IEEE Online Journal of Surveys and
Tutorials (2002-07). He is currently the coding and communication theory area editor for IEEE Trans. on Communications (2011-present) and an editor for Elsevier Physical Communications Journal (2010-present).
\end{IEEEbiographynophoto}

\end{document}